\documentclass{amsart}
\usepackage{amsfonts}
\usepackage{amsmath}
\usepackage{amssymb}
\usepackage{mathtools}
\usepackage{color}
\usepackage{amsthm}
\usepackage{ifpdf}
\usepackage{todonotes}
\usepackage{graphicx, caption}
\usepackage{bm}
\usepackage{pdflscape,afterpage}
\ifpdf
  \usepackage[pdftex,
  	bookmarks=true,
	unicode=true,
	pdfborder={0 0 0},
	pdfstartview={FitV},
	pdfpagelayout={TwoPageRight},
	pdftitle={},
	pdfauthor={Christian Bayer, Peter Friz, Archil Gulisashvili, Blanka
          Horvath, Benjamin Stemper},
	colorlinks=false,
	linktoc=all,
	plainpages=false,
	pdfpagelabels,
	pdffitwindow=true]{hyperref}
\else
  \usepackage[dvips]{hyperref}
\fi

\theoremstyle{plain}
\newtheorem{theorem}{Theorem}[section]
\newtheorem{corollary}[theorem]{Corollary}
\newtheorem{lemma}[theorem]{Lemma}
\newtheorem{proposition}[theorem]{Proposition}

\theoremstyle{definition}
\newtheorem{assumption}[theorem]{Assumption}
\newtheorem{definition}[theorem]{Definition}
\newtheorem{example}[theorem]{Example}

\theoremstyle{remark}
\newtheorem{remark}[theorem]{Remark}

\newcommand{\R}{\mathbb{R}}

\newcommand{\f}[2]{\frac{#1}{#2}}
\newcommand{\pa}{\partial}
\newcommand{\half}{\f{1}{2}}

\newcommand{\abs}[1]{\left\lvert#1\right\rvert} 
\newcommand{\norm}[1]{\left\lVert#1\right\rVert} 
\newcommand{\indic}[1]{\mathbf{1}_{#1}} 
\newcommand{\ip}[2]{\left\langle #1\,,#2\right\rangle} 

\newcommand{\barrho}{\overline{\rho}}

\renewcommand{\bar}{\overline}
\renewcommand{\tilde}{\widetilde}
\renewcommand{\hat}{\widehat}
\newcommand{\hatepsilon}{\hat{\varepsilon}}

\DeclareMathOperator{\id}{id}
\newcommand{\fs}{\mathfrak{s}}
\newcommand{\const}{\mathrm{const}}
\newcommand{\Var}{\operatorname{Var}}
\newcommand{\Cov}{\operatorname{Cov}}
\newcommand{\etalchar}[1]{$^{#1}$}

\newcommand{\rtheta}{\theta}

\numberwithin{equation}{section}

\title[Short-time near-the-money skew in RFV models]{Short-time near-the-money skew in rough fractional volatility models}

\author{C. Bayer, P. K. Friz, A. Gulisashvili, B. Horvath, B. Stemper}
\address{WIAS Berlin, TU and WIAS Berlin, Ohio University, Imperial College London, 
TU and WIAS Berlin}
\email{christian.bayer@wias-berlin.de, friz@math.tu-berlin.de, gulisash@ohio.edu, 
b.horvath@imperial.ac.uk, stemper@math.tu-berlin.de}

\date{\today}

\begin{document}

\begin{abstract}
  We consider rough stochastic volatility models where the driving noise of volatility has fractional scaling, in the "rough'' regime of Hurst parameter $H < 1/2$. This regime recently attracted a lot of attention both from the statistical and option pricing point of view.  With focus on the latter, we sharpen the  large deviation results of Forde-Zhang (2017) in a way that allows us to zoom-in around the money while maintaining full analytical tractability. More precisely, this amounts to proving higher order moderate deviation estimates, only recently introduced in the option pricing context. This in turn allows us to push the applicability range of known at-the-money skew approximation formulae from CLT type log-moneyness deviations of order $t^{1/2}$ (recent works of Al\`{o}s, Le\'{o}n \& Vives and Fukasawa) 
  to the wider moderate deviations regime.
\end{abstract}

\thanks{We gratefully acknowledge financial support through DFG research
  grants FR2943/2 and BA5484/1 (C. Bayer, P.K. Friz, B. Stemper),  European Research Council Grant CoG-683166 (P.K. Friz), and SNF Early
  Postdoc Mobility Grant 165248 (B. Horvath) respectively.}
 \keywords{rough stochastic volatility model, European option pricing, small-time asymptotics, moderate deviations}
 \subjclass[2010]{91G20, 60H30, 60F10, 60H07, 60G22, 60G18}

\maketitle

\tableofcontents


\section{Introduction}
\label{sec:intr-notat}
Since the groundbreaking work of Gatheral, Jaisson and Rosenbaum
\cite{GJR14}, the past two years have brought about a gradual shift in
volatility modeling, leading away from classical diffusive stochastic
volatility models towards so-called rough volatility models. The term was coined in
\cite{GJR14} and \cite{BFG16}, and it essentially describes a family of
(continuous-path) stochastic volatility models where the driving noise of the volatility process
has H\"{o}lder regularity lower than Brownian motion, typically achieved by
modeling the fundamental noise innovations of the volatility process as a
fractional Brownian motion with Hurst exponent (and hence H\"{o}lder
regularity) $H<1/2$. Here, we would also like to mention
  pioneering work on asymptotics for rough volatility models in ~\cite{ALV07} and~\cite{Fuk11}.  A major appeal of such rough
volatility models lies in the fact that they effectively capture several
stylized facts of financial markets both from a statistical \cite{GJR14,
  BLP16} and an option-pricing point of view \cite{BFG16}.
In particular, with regards to the latter point of view, a widely observed empirical
phenomenon in equity markets is the ``steepness of the smile on the
short end'' describing the fact that as time to maturity becomes small the
empirical implied volatility skew follows a power law with negative exponent,
and thus becomes arbitrarily large near zero. While standard stochastic
volatility models with continuous paths struggle to capture this phenomenon,
predicting instead a constant at-the-money implied volatility behaviour on the
short end \cite{Gat11}, models in the fractional stochastic volatility family
(and more specifically so-called rough volatility models) constitute a class,
well-tailored to fit empirical implied volatilities for short dated options.

Typically, the popularity of asset pricing models hinges on the availability
of efficient numerical pricing methods. In the case of diffusions, these
include Monte Carlo estimators, PDE discretization schemes, asymptotic
expansions and transform methods.  With fractional Brownian motion being the
prime example of a process beyond the semimartingale framework, most currently
prevalent option pricing methods -- particularly the ones assuming
semimartingality or Markovianity -- may not easily carry over to the rough
setting.  In fact, the memory property (aka non-Markovianity) of fractional
Brownian motion rules out PDE methods, heat kernel methods and all related
methods involving a Feynman-Kac-type Ansatz.  Previous work has thus focused
on finding efficient Monte Carlo simulation schemes
\cite{BFG16,BLP15,BFGMS17} or -- in the special case of the Rough Heston
model -- on an explicit formula for the characteristic function of the log-price (see \cite{ER16}), 
thus in this particular model making pricing amenable to Fourier
based methods.  In our work, we rely on small-maturity approximations of option prices.
This is a well-studied topic. See, e.g., \cite{ALV07, GVZ15} (the at-the-money (ATM) regime) 
or \cite{DFJV14a, DFJV14b, GJR14b, GHJ16, GVZ15} (the out-of-the-money (OTM) regime, where
large deviations results are used). We also refer the reader to the papers
\cite{Fuk11, Fuk17, FZ17} concerning large deviations, and to \cite{MT16, Osa07, MS03,
  MS07} for related work. Based on the moderate deviations regime, 
Friz et al. \cite{FGP17} have recently introduced another
regime called moderately-out-of-the-money (MOTM), which, in a sense, effectively navigates
between the two regimes mentioned above, by rescaling the strike with respect to the
time to maturity. This approach has various advantages. On the one hand, it
reflects the market reality that as time to maturity approaches zero, strikes
with acceptable bid-ask spreads tend to move closer to the money (see
\cite{FGP17} for more details). On the other hand, it allows us to zoom in on the 
term structure of implied volatility around the money at a high resolution scale.
To be more specific, our paper adds to the existing literature in two ways. First,
we obtain a generalization of the Osajima energy expansion \cite{Osa15} to a
non-Markovian case, and using the new expansion, we extend the analysis of
\cite{FGP17} to the case, where the volatility is driven by a rough $(H < 1/2)$
fractional Brownian motion. Indeed, Laplace approximation methods on Wiener space in the spirit of
Ben Arous \cite{BA88} and Bismut \cite{Bis84} remain valid in the present
context, and our analysis builds upon this framework in a fractional
setting. Second, we use an asymptotic expansion going back to Azencott
\cite{Aze85} to bypass the need for deriving an
asymptotic expansion of the density of the underlying process to obtain
asymptotics for option prices. We display the potential prowess of this
approach by applying it to our specific model, and derive asymptotics for call
prices directly, irrespectively of corresponding density asymptotics.
Finally, using a version of the "\emph{rough Bergomi model}" \cite{BFG16}, we
demonstrate numerically that our implied volatility asymptotics capture very
well the geometry of the term structure of implied volatility over a wide
array of maturities, extending up to a year.

The paper is organized as follows: In Section \ref{sec:ExpositionAssumptions}
we set the scene, describing the class of models included in our framework
(\eqref{eq:SDEoriginal} and \eqref{eq:defBhat}) and recalling some known
results (\eqref{eq:RatefunctionfBmtriplet} and
\eqref{eq:RatefunctionFordeZhang}), which are the starting point of our
analysis. Most importantly, we argue that for small-time considerations it
would suffice to restrict our attention to a class of stochastic volatility
models of the form \eqref{eq:LogpriceDrift} with a volatility process
driven by a Gaussian Volterra process such as in \eqref{eq:defBhat}. We
formulate general assumptions on the Volterra kernel (Assumptions \ref{ass:1}
and \ref{ass:conditions}) and on the function $\sigma$ in
\eqref{eq:LogpriceDrift} (Assumption~\ref{ass:sigma}) under which our
results are valid. In Section \ref{sec:main-results} we gather our main
results, concerning a higher order expansion of the energy (Theorem
\ref{thr:main-energy}), and a general expansion formula for the corresponding
call prices. We derive the classical Black-Scholes expansion for the call
price, using the latter result mentioned above.  In addition, in Section
\ref{sec:main-results} we formulate moderate deviation expansions, which allow
us to derive the corresponding asymptotic formulae for implied volatilities
and implied volatility skews.  Finally, Section \ref{sec:SimulationResults}
displays our simulation results.  Sections \ref{sec:energy-expansion},
\ref{sec:price-expansion} and \ref{sec:moderate-deviations} are devoted to
proofs of the energy expansion, the price expansion and the moderate
deviations expansion, respectively.  In the appendix, we have collected some
auxiliary lemmas, which are used in different sections.

\section{Exposition and assumptions}
\label{sec:ExpositionAssumptions}

We consider a rough stochastic volatility model, normalized to $r=0$ and
$S_0 = 1$, of the form suggested by Forde-Zhang~\cite{FZ17}
\begin{equation}\label{eq:SDEoriginal}
\frac{dS_{t}}{S_{t}}=\sigma (\widehat{B}_t)d\left( \bar{\rho}W_t+\rho B_t\right) .
\end{equation}
Here $\left( W,B\right) $ are two independent standard Brownian motions, $\rho
\in(-1,1)$ a correlation parameter, and
$\bar{\rho}^{2}=1-\rho ^{2}$. Then $\bar{\rho} W+\rho B $ is another standard
Brownian motion which has constant correlation $\rho $ with the factor $B$,
which drives the stochastic volatility $$\sigma _{\mathrm{stoch}}\left( t,\omega
\right) :=\sigma (\widehat{B}_{t}\left( \omega \right) )\equiv \sigma
(\widehat{B}).$$ Here $\sigma(.)$ is some real-valued function, typically smooth but not bounded, and we will denote by $\sigma_0:=\sigma(0)$ the spot
volatility, with $\widehat{B}$ a Gaussian (Volterra) process of the
form
\begin{equation}\label{eq:defBhat}
\widehat{B}_{t} =\int_{0}^{t}K\left( t,s\right) dB_{s}, \quad t \ge 0,
\end{equation}
for some kernel $K$, which shall be further specified in Assumptions
\ref{ass:1} and \ref{ass:conditions} below. The log-price
${X}_{t}=\log \left( S_{t}\right) $ satisfies
\begin{equation}\label{eq:LogpriceDrift}
dX_t= -\frac{1}{2}\sigma ^{2}(\widehat{B}_t)dt+\sigma
(\widehat{B}_t)d\left( \bar{\rho}W+\rho B\right),
\quad X_{0}=0.
\end{equation}

Recall that by Brownian scaling, for fixed $t>0$, 
 \begin{equation*}
 ( B_{ts}, W_{ts})_{s\geq 0}
 \overset{law}{=} \varepsilon( B_{s}, W_{s})_{s\geq 0} , \text{ \ \ where \ \ } \varepsilon \equiv \varepsilon (t) \equiv t^{1/2}.
 \end{equation*}
 
\noindent As a direct consequence, classical short-time SDE problems can be analyzed as small-noise problems on a unit time horizon. For our analysis, it will also be crucial to impose such
a scaling property on the Gaussian process $\widehat{B}$ (more precisely, on
the kernel $K$ in \eqref{eq:defBhat}) driving the volatility process in our
model:
\begin{assumption}[Small time self-similarity]
  \label{ass:1}
  There exists a number $t_0$ with $0 < t_0 \leq 1$ and a function $t \mapsto \hatepsilon = \hatepsilon(t)$,
  $0 \leq t \leq t_0$, such that 
 \begin{equation*}
  ( \widehat{B}_{ts}: 0 \leq s \leq t_0 ) 
   \overset{law}{=} ( \hat{\varepsilon}\widehat{B}_s : 0 \leq s \leq t_0 ).
  \end{equation*}
\end{assumption}

In fact, we will always have 
$$\hat{\varepsilon}  \equiv  \hat{\varepsilon} (t) \equiv  t^{H}=\varepsilon ^{2H},$$
which covers the examples of interest, in particular standard fractional
Brownian motion $\widehat{B}=B^{H}$ or Riemann-Liouville fBM with explicit
kernel $K\left( t,s\right) = \sqrt{2H}\abs{t-s}^{H-1/2}$.  (This is very
natural, even from a general perspective of self-similar processes, see
\cite{Lam62}.)

We insist that no (global) self-similarity of $\widehat{B}$ is required, as
only $\widehat{B}|_{[0,t]}$ for arbitrarily small $t$ matters. 

\begin{remark}
  It should be possible to replace the fractional Brownian
  motion by a certain fractional Ornstein-Uhlenbeck process in the results
  obtained in this paper. Intuitively, this replacement creates a negligible
  perturbation (for $t \ll 1$) of the fBm environment.  A similar situation
  was in fact encountered in \cite{CF10}, where fractional scaling at times
  near zero was important. To quantify the perturbation, the authors of
  \cite{CF10} introduced an easy to verify coupling condition (see Corollary 2
  in \cite{CF10}). It should be possible to employ a version of this condition
  in the present paper to justify the replacement mentioned above. We will
  however not pursue this point further here.
\end{remark}

\begin{remark}
  Throughout this article, one can consider a classical (Markovian, diffusion)
  stochastic volatility setting by taking $K \equiv 1$, or equivalently $H
  \equiv 1/2$, by simply ignoring all hats ( $\widehat{\cdot }$ ) in the
  sequel. In particular then, $\frac{\hat{\varepsilon}}{\varepsilon} \equiv 1$
  in all subsequent formulae.
\end{remark}

General facts on large deviations of Gaussian measures on Banach spaces
\cite{DS89} such as the path space $C([0,1],\mathbb{R}^3)$ imply that a large
deviation principle holds for the triple
$\{\hat{\varepsilon}(W,B,\widehat{B}): \hat{\varepsilon}>0\}$, with speed
$\hat{\varepsilon}^2$ and rate function
\begin{equation}\label{eq:RatefunctionfBmtriplet}
\begin{cases}
\frac{1}{2}\left\Vert h\right\Vert _{H_0^{1}}^{2}+\frac{1}{2}\left\Vert
f\right\Vert _{H_0^{1}}^{2}, & f,h\in H_0^{1}\text{ and }\widehat{f}=K\dot{f},\\
+\infty, & \text{otherwise},
\end{cases}
\end{equation}
where 
\begin{align*}
K\dot{f}(t) \coloneqq\int_{0}^{t}K\left(
t,s\right) \dot{f}(s) ds
\end{align*}
for $f \in H_0^1$, the space of absolutely continuous paths with $L^2$ derivative
\begin{equation}\label{eq:H01}
	H_0^1 \coloneqq \left\{ f: [0,1] \to \R \text{ continuous }\, \left|
            \, \norm{f}_{H^1_0}^2 \coloneqq \int_0^1 \abs{\dot{f}(s)}^2 ds <
            \infty, \ f(0) = 0 \right.\right\}. 
\end{equation}
This enables us to derive a large deviations principle for $X$ in \eqref{eq:LogpriceDrift}:
the (local) small-time self-similarity property of $\widehat{B}$ (Assumption \ref{ass:1})   implies that $X_{t}\overset{law}{=}X_{1}^{\varepsilon
} $ where
\begin{equation*}
  dX^{\varepsilon }_t=\sigma (\hat{\varepsilon}\widehat{B}_t)\varepsilon d\left(
    \bar{\rho}W_t+\rho B_t\right) - \half \varepsilon^2
  \sigma^2(\hat{\varepsilon} \hat{B}_t) dt,\quad X_{0}^{\varepsilon }=0.
\end{equation*}
For what follows, it will be convenient to consider a rescaled version of \eqref{eq:LogpriceDrift}
\begin{equation*}
  d\widehat{X}^{\varepsilon }_t\equiv d\left(
    \frac{\hat{\varepsilon}}{\varepsilon } 
    X^{\varepsilon }_t\right) =\sigma (\hat{\varepsilon}\widehat{B}_t)\hat{\varepsilon}
  d\left( \bar{\rho}W_t+\rho B_t\right)
  -\half \varepsilon \hat{\varepsilon} \sigma^2(\hat{\varepsilon} \hat{B}_t) dt
 ,\quad\widehat{X}_{0}^{\varepsilon }=0.
\end{equation*}
Under a linear growth condition on the function $\sigma$, Forde--Zhang \cite{FZ17} 
use the extended contraction principle to establish a large deviations principle for ($\widehat{X}^{\varepsilon }_1$) with speed 
$\hat{\varepsilon}^2$.  More precisely, with
\begin{equation}\label{eq:defItotypeMap}
\varphi _{1}\left( h,f\right) \coloneqq \Phi_1
(h,f,\widehat{f})=\int_{0}^{1}\sigma (\widehat{f})d\left( \bar{\rho}h+\rho
  f\right),
\end{equation}
the rate function is given by
\begin{align}\label{eq:RatefunctionFordeZhang}
\begin{split}
I\left( x\right) &=\inf_{h,f\in H_0^{1}}\left\{ \frac{1}{2}\int_{0}^{1}\dot{h}%
^{2}dt+\frac{1}{2}\int_{0}^{1}\dot{f}^{2}dt:\varphi _{1}\left( h,f\right)
=x\right\} \\
&=\inf_{f\in H_0^{1}}\left\{ \frac{1}{2}\frac{\left( x-\rho \left\langle
\sigma (\widehat{f}),\dot{f}\right\rangle \right)^{2}}{\bar{\rho}^{2}\left\langle
\sigma ^{2}(\widehat{f}),1\right\rangle }+ \half \int_{0}^{1}\dot{f}^{2}dt\right\},
\end{split}
\end{align}
where $\ip{\cdot}{\cdot}$ denotes the inner product on
  $L^2\left( [0,1], dt \right)$.
Several other proofs (under varying assumptions on $\sigma$) have appeared
since \cite{JPS17, BFGMS17,Gul17}.

\medskip 
As a matter of fact, this paper relies on moderate - rather than large - deviations, as emphasized in (iiic) below. To this end, let us make 

\begin{assumption}  \label{ass:sigma}
  \leavevmode
  \begin{enumerate}
  \item[(i)] (Positive spot vol) Assume $\sigma :\mathbb{R} \rightarrow \mathbb{R}$ is smooth with $\sigma_0 := \sigma(0)>0$. 
  \item[(ii)] (Roughness) The Hurst parameter $H$ satisfies $H \in (0, 1/2]$.
  \item[(iiia)] (Martingality) The price process $S = \exp X$ is a martingale. 
  \item[(iiib)] (Short-time moments) $\forall m<\infty \ \exists t>0$: $\ E(S_t^m) < \infty$.
\end{enumerate}
\end{assumption}

\medskip
While condition (iiia) hardly needs justification, we emphasize that conditions (iiia-b) are only used to the extent that they imply condition (iiic) given below (which thus may replace (iiia-b) as an alternative, if more technical, assumption). The reason we point this out explicitly is that all the conditions (iiia-c) are implicit (growth) conditions on the function $\sigma(.)$. For instance, (iiia-b) was seen to hold under a linear growth assumption \cite{FZ17}, whereas the log-normal volatility case (think of $\sigma(x) = e^x$) is complicated. Martingality, for instance, requires $\rho \le 0$ and there is a critical moment $m^* = m^*(\rho)$, even when $\rho < 0$. See \cite{sin1998,Jourdain, LM07} for the case $H=1/2$ and the forthcoming work \cite{FG18} for the general rough case $H\in(0,1]$.  We view (iiic) simply as a more flexible condition that can hold in situations where (iiib) fails.

\medskip

\begin{enumerate}
\item[(iiic)] (Call price upper moderate deviation bound) 
For every $\beta \in (0,H)$,  and every fixed $x>0$, and $\hat x_\varepsilon := x \varepsilon^{1-2H+2\beta}$,
  $$
          E [  (e^{X^\varepsilon_1} - e^{\hat x_\varepsilon })^+ ]    
      \le \exp \left( - \frac{x^2 + o(1)}{2 \sigma_0^2 \varepsilon^{4H - 4 \beta}}  \right).
 $$  
\end{enumerate}
This condition is reminiscent of the ``upper part'' of the large deviation estimate obtained in \cite{FZ17}
\begin{equation}
            E [  (e^{X^\varepsilon_1} - e^{ x \varepsilon^{1-2H} })^+ ]    
      = \exp \left( - \frac{I(x) + o(1)}{\varepsilon^{4H}}  \right) \ .    \label{equ:LDPu}
\end{equation} 
If fact, if one {\it formally} applies this with $x$ replaced by $ x\varepsilon^{2 \beta}$, followed by Taylor expanding the rate function, 
$$
I (x \varepsilon^{2 \beta}) \sim \tfrac{1}{2} I''(0) x^2 \varepsilon^{4 \beta} = \tfrac{1}{2 \sigma_0^2} x^2 \varepsilon^{4 \beta} \ ,
$$
one readily arrives at the estimate (iiic). Unfortunately, $o(1)=o_x(1)$ in (\ref{equ:LDPu}), which is a serious obstacle in making this argument rigorous. Instead, we will give a direct argument (Lemma \ref{lem:iiiabc}) to see how (iiia-b) implies (iiic).

\medskip 

In the sequel, we will use another mild assumption on the kernel.

\begin{assumption}
 \label{ass:conditions}
 The kernel $K$ has the following properties
 \begin{itemize}
 \item[(i)] $\widehat{B}_t = \int_0^t K(t,s) dB_s$ has a continuous (in $t$)
   version on $[0,1]$.
 \item[(ii)] $\forall t \in[0,1]: \ \int_0^t K(t,s)^2 ds < \infty$.
 \end{itemize}
\end{assumption}
Note that the Riemann-Liouville kernel $K(t,s) = \sqrt{2H} (t-s)^\gamma$, $\gamma
=H - 1/2$ satisfies Assumption~\ref{ass:conditions}.

  \begin{remark}
    Assumption~\ref{ass:conditions} implies that the Cameron-Martin space
    $\mathcal{H}$ of $\widehat{B}$ is given by the image of $H^1_0$ under $K$, i.e.,
   \begin{equation*}
     \mathcal{H} = \{K \dot{f} \mid f \in H^1_0\}.
   \end{equation*}
   See Lemma~\ref{lem:cameron-martin-space} and Remark~\ref{rem:assumption-K}
   for more details. A reference and also a sufficient condition for
   Assumption~\ref{ass:conditions} (i) can be found e.g. in~\cite[Section 3]{Dec05}.
  \end{remark}


\section{Main results}
\label{sec:main-results}

The following result can be seen as a non-Markovian extension of work by Osajima
\cite{Osa15}. The statement here is a combination of
Theorem~\ref{thr:smoothness-energy-general} and
Proposition~\eqref{prop:energy-expansion-general} below.  Recall that $\sigma
_{0}=\sigma \left( 0\right) $ represents spot-volatility. We also set $\sigma
_{0}^{\prime }\equiv \sigma ^{\prime }\left( 0\right) $.

\begin{theorem}[Energy expansion]
 \label{thr:main-energy}
 The rate function (or energy) $I$ in \eqref{eq:RatefunctionFordeZhang} is
 smooth in a neighbourhood of $x=0$ (at-the-money) and it is of the form
\begin{equation*}
  I\left( x\right) =\frac{1}{\sigma _{0}^{2}}\frac{x^{2}}{2}-\left( 6 \rho
    \f{\sigma_0^\prime}{\sigma_0^4} \int_0^1 \int_0^t K(t,s) ds dt \right)
  \frac{x^{3}}{3!}+ \mathcal{O}(x^4).
\end{equation*}
\end{theorem}

The next result is an exact representation of call prices, valid in a
non-Markovian generality, and amenable to moderate- and large-deviation
analysis (Theorem \ref{thm:mod} below) as well as to full asymptotic
expansions, which will be explored in forthcoming work.

\begin{theorem}[Pricing formula] \label{thr:main-price=expansion} For a fixed
  log-strike $x\geq 0$ and time to maturity $t>0$, set
  $\widehat{x}:=\f{\varepsilon}{\hat{\varepsilon}}x$, where $\varepsilon
  =t^{1/2}$ and $\hat{\varepsilon} =t^{H}=\varepsilon ^{2H}$, as before. Then
  we have
  \begin{eqnarray}
    c(\widehat{x},t) &=&E\left[\left( \exp \left( X_{t}\right) -\exp
        \widehat{x}\right) ^{+} \right] \nonumber \\ 
    &=&e^{-\frac{I\left( x\right) }{\hat{\varepsilon}^{2}}}e^{\frac{\varepsilon 
      }{\hat{\varepsilon}}x}\,J\left( \varepsilon ,x\right),    \label{equ:calland}
  \end{eqnarray}%
  where
  \begin{equation*}
    J\left( \varepsilon ,x\right) :=E\left[ e^{-\frac{I^{\prime }\left( x\right) 
        }{\hat{\varepsilon}^{2}}\widehat{U}^{\varepsilon }}\left( \exp \left( \tfrac{%
            \varepsilon }{\hat{\varepsilon}}\widehat{U}^{\varepsilon }\right) -1\right)
      e^{I^{\prime }\left( x\right) R^\varepsilon_{2}} \ 
      \indic{\widehat{U}^{\varepsilon }\geq 0}\right]
  \end{equation*}%
  and $\widehat{U}^{\varepsilon}$ is a random variable of the form
  \begin{align}\label{eq:U}
    \widehat{U}^{\varepsilon }=\hat{\varepsilon}g_{1} +\hat{\varepsilon}^{2}R^\varepsilon_{2}
  \end{align}
  with $g_{1}$ a centred Gaussian random variable, explicitly given in
  equation (\ref{g1Expl}) below, and $R^\varepsilon_2$ is a (random) remainder term, in the
  sense of a stochastic Taylor expansion in $\hat{\varepsilon}$, see
  Lemma~\ref{Lem:STEforZ} for more details.
\end{theorem}

\begin{example}[Black-Scholes model] \label{ex:BS-example}
We fix volatility $\sigma \left( \cdot \right) \equiv
\sigma >0$, and $H=1/2$ so that $\hat{\varepsilon}=\varepsilon$ and all $\hat\cdot$ can be omitted. Energy is given by $I\left( x\right) =\frac{
x^{2}}{2\sigma ^{2}}$ and
$$
         U^{\varepsilon }= {\varepsilon}g_{1} + {\varepsilon}^{2}R^\varepsilon_{2} \equiv {\varepsilon} \sigma W_1 - {\varepsilon}^{2} \sigma^2 / 2
$$
with $R^\varepsilon_{2} = R_{2} \equiv -\sigma^2/2$ independent of $\varepsilon$. Moreover,
\begin{eqnarray}
J\left( \varepsilon ,x\right)  &=&
E\left[ e^{-\frac{I^{\prime }\left( x\right) 
        }{{\varepsilon}^{2}}U^{\varepsilon }}\left( e^{ U^{\varepsilon }} -1\right)
      e^{I^{\prime }\left( x\right) R_{2}} \ 
      \indic{U^{\varepsilon }\geq 0}\right]
\nonumber \\
 &=&E\left[ e^{-\frac{I^{\prime }\left(
x\right) }{\varepsilon }g_{1}}\left( e^{\varepsilon g_{1}-\varepsilon^2 \tfrac{\sigma^2}{2}}-1\right) 
\ \indic{\{g_{1}\geq
 \tfrac{ \varepsilon \sigma^2}{2}\}}\right]  \nonumber \\
  &=&E\left[ e^{-\alpha W_1}\left( e^{\varepsilon\sigma W_1- \tfrac{(\varepsilon\sigma)^2}{2}}-1\right) 
\ \indic{\{ W_1\geq
 \tfrac{ \varepsilon\sigma}{2}\}}\right]  \nonumber \\
&=&e^{-\tfrac{(\varepsilon \sigma)^2}{2}}
M\left( -\alpha%
+\varepsilon \sigma \right) 
- M\left( -\alpha\right)    \label{equ:JinBS}
\end{eqnarray}%
with $\alpha:=\frac{I^{\prime }\left(x\right) \sigma }{\varepsilon } = \frac{1}{\sigma} (x/ \varepsilon)$, and, in terms of the standard Gaussian cdf $\Phi$,
$$
M\left( \beta \right):= E\left[ e^{\beta W_1}
\ \indic{\{ W_1\geq
 \tfrac{ \varepsilon\sigma}{2}\}}\right]  = e^{\beta ^{2}/2} \Phi \left(  \beta -\tfrac{\varepsilon\sigma}{2} \right) \ .
 $$
Using the expansion $\Phi( -y ) = \frac{1}{y \sqrt{2\pi}}e^{-y^2/2}(1 - y^{-2}+...)$, as $y \to \infty$ one deduces, 
for fixed $x>0$, 
the asymptotic relation, as $\varepsilon \to 0$, 
\begin{equation}
J\left( \varepsilon ,x\right)  \sim 
\frac{e^{-x/2}}{\sqrt{2\pi }}\frac{\varepsilon ^{3}\sigma ^{3}}{x^{2}}.
\end{equation}
We will be interested (cf. Theorem \ref{thm:mod}) in replacing $x$ by $\tilde x = x \varepsilon^{2 \beta} \to 0$ for $\beta > 0$. This gives 
$\tilde \alpha = \frac{1}{\sigma}( x / \varepsilon^{1-2\beta}) $ 
and the above analysis, now based on $\tilde \alpha \to \infty$, remains valid\footnote{More terms in the expansion of $\Phi$ are needed.} for $\beta$ in the ``moderate'' regime $\beta \in [0,1/2)$ and we obtain
\begin{equation}
\forall x>0, \beta \in [0,1/2): 
J\left( \varepsilon ,x \varepsilon^{2 \beta} \right)  \sim 
\frac{1}{\sqrt{2\pi }}\frac{\varepsilon ^{3 - 4 \beta}\sigma ^{3}}{x^{2}}.    \label{JBSmode}
\end{equation}
Let us point out, for the sake of completeness, that a similar expansion is {\it not} valid for $\beta > 1/2$. 
To see this, first note that  (\ref{equ:calland}) implies that $J(\varepsilon, x)|_{x=0}$ is precisely the ATM call price
  with time $t=\varepsilon^2$ from expiration. Well-known ATM asymptotics then imply 
  that $J(\varepsilon, x)|_{x=0} \sim \tfrac{1}{\sqrt{2\pi}} \varepsilon \sigma$ as $\varepsilon \to 0$. 
  These asymptotics are unchanged in case of $o(t^{1/2})=o(\varepsilon)$ out-of-moneyness (``almost-at-the-money'' in the terminology of \cite{FGP17}), which readily implies
$$
\forall x>0, \beta > 1/2 : 
J\left( \varepsilon ,x \varepsilon^{2 \beta} \right)  \sim 
\frac{1}{\sqrt{2\pi }} \varepsilon \sigma = \text{const} \times \varepsilon
$$
At last, we have the borderline case $\beta=1/2$, or $\tilde x = x \varepsilon$. From e.g. \cite[Thm 3.1]{muhle-karbe2011}, we see that $c (x \varepsilon, \varepsilon^2) \sim a(x;\sigma) \varepsilon$ with positive constant $a(x;\sigma)$. 
A look at (\ref{equ:calland}) then reveals
$$
      \forall x>0 : J\left( \varepsilon ,x \varepsilon \right)  \sim a(x; \sigma) \varepsilon e^\frac{x^2}{2\sigma^2} = \text{const} \times \varepsilon \ .
$$
%
%

 \noindent For the call price expansion in the large / moderate deviations regime, $\beta \in [0,1/2)$, the polynomial in $\varepsilon$-behaviour of (\ref{JBSmode}) implies
 that the $J$-term in the pricing formula will be negligible on the moderate
 / large deviation scale, in the sense for any $\theta>0$, we have $\varepsilon^\theta \log J(\varepsilon,x \varepsilon^{2\beta}) \to 0$ 
 as $\varepsilon \to 0$. Consequently, with $k_t = k t^\beta$, for $t=\varepsilon^2$, $k>0$,
 $\beta \in [0,1/2)$, we get the ``moderate" Black-Scholes call price expansion,

 \begin{equation*}
-\log c_{BS} (k_{t},t)=\frac{1}{ t^{1-2\beta }}%
\frac{k^{2}}{2\sigma^2}\left( 1+o\left( 1\right) \right) \text{ as }t\downarrow 0.
\end{equation*}%
\end{example}
\noindent While the above can be confirmed by elementary analysis of the Black--Scholes formula, the following theorem exhibits it 
as an instance of a general principle. See \cite{FGP17} for a general diffusion statement.\\

\begin{theorem}[Moderate Deviations] \label{thm:mod}
In the rough volatility regime $H\in (0,1/2]$,
consider log-strikes of the form
\begin{equation*}
k_{t}=kt^{\frac{1}{2}-H+\beta } \quad \textrm{for a constant} \quad k \geq 0.
\end{equation*}%
(i)\ For $\beta \in ( 0,H)$, and every $\theta > 0$, we have%
\begin{equation*}
-\log c(k_{t},t)=\frac{I^{\prime \prime }\left( 0\right) }{t^{2H-2\beta }}%
\frac{k^{2}}{2}    + O(t^{3\beta - 2H}) + O(t^{-\theta})
\ \ \ \text{ as }t\downarrow 0.
\end{equation*}%
(ii) For $\beta \in ( 0,\frac{2}{3}H)$, and every $\theta > 0$, we have%
\begin{equation*}
-\log c(k_{t},t)=\frac{I^{\prime \prime }\left( 0\right) }{t^{2H-2\beta }}
\frac{k^{2}}{2}+\frac{I^{\prime \prime \prime }\left( 0\right) }{
t^{2H-3\beta }}\frac{k^{3}}{6}  + O(t^{4\beta - 2H})+O( t^{-\theta}) \ \ \ \text{ as }t\downarrow 0.
\end{equation*}%
Moreover,
\begin{align*}
I^{\prime \prime} \left( 0\right)  &= \frac{1}{\sigma _{0}^{2}}, \\
I^{\prime \prime \prime }\left( 0\right)  &= -6 \rho
    \f{\sigma_0^\prime}{\sigma_0^4} \int_0^1 \int_0^t K(t,s) ds dt
    =-6 \rho    \f{\sigma_0^\prime}{\sigma_0^4}\langle K 1, 1\rangle,
\end{align*}
where $\ip{\cdot}{\cdot}$ is the inner product in $L^2\left([0,1]\right)$.
\end{theorem}
\begin{remark}
In principle, further terms (of order $t^{i\beta-2H}$) can be added to this
expansion of log call prices, given that the energy has sufficient regularity,
see Theorem \ref{T:stra}. 
We also note that, for  small enough $\beta$, the error term $O( t^{-\theta})$ can be omitted. In any case, 
one can replace the additive error bounds by (cruder) ones, where the right-most term in the expansion is multiplied with $(1+o(1))$, as was
done in \cite{FGP17}.  
\end{remark}
\begin{proof}[Proof of Theorem \ref{thm:mod}]
We apply Theorem~\ref{thr:main-price=expansion} with
    $\hat{x} = k_t = k t^{1/2 - H + \beta}$, i.e., with $x = k t^\beta = k \varepsilon^{2\beta}$. In particular,  we so get, with $\hat \varepsilon = t^H$ and $\varepsilon = t^{1/2}$,
    $$
    c(k_t,t) = e^{-\frac{I\left( x\right) }{\hat{\varepsilon}^{2}}}e^{\frac{\varepsilon 
      }{\hat{\varepsilon}}x}\,J\left( \varepsilon , k \varepsilon^{2\beta} \right).
    $$
    The technical Proposition~\ref{prop:J-bounds} asserts that, for fixed $k>0$, the factor $J$ is negligible in the sense that, for every $\theta > 0$,
 $$
                          \varepsilon^{\theta} \log J ( \varepsilon, k \varepsilon^{2\beta} ) \to 0  \ \ \ \text{ as $\varepsilon \to 0$} \ .     
$$     
%
The theorem now follows immediately from the Taylor expansion of $I(x)$ around $x=0$ (see
    Theorem~\ref{thr:main-energy}), plugging in $x = k t^\beta$.  Indeed, replacing $I(x)$ by the Taylor-jet seen in (i),(ii), leads
    exactly to an error term $O(t^{3\beta - 2H})$, resp. $O(t^{4\beta - 2H})$ .
\end{proof}

	Fix real numbers $k> 0$, $0< H<\frac{1}{2}$, $0<\beta< H$, and an integer $n\ge 2$. For every $t> 0$, set 
	$$
	k_t=kt^{\frac{1}{2}-H+\beta},
	$$
	and denote 
        \begin{equation*}
	\phi_{n,H,\beta, \rtheta}(t)=\max\left\{t^{2H-2\beta-\rtheta}
          ,t^{(n-1)\beta}\right\}.
      \end{equation*}
        Here, $\theta>0$ can be arbitrarily small.
	It is clear that for all small $t$ and $\theta$ small
          enough,  
        \begin{equation*}
	\phi_{n,H,\beta,\rtheta}(t)=t^{2H-2\beta-\rtheta}
        \Leftrightarrow 2H-2\beta\le(n-1)\beta
	\Leftrightarrow\frac{2H}{n+1}\le\beta,
      \end{equation*}
	while
	$$
	\phi_{n,H,\beta,\rtheta}(t)=t^{(n-1)\beta}\Leftrightarrow 2H-2\beta>(n-1)\beta
	\Leftrightarrow\beta<\frac{2H}{n+1}.
	$$ 
	
	The following statement provides an asymptotic formula for the implied variance.
	\begin{theorem}\label{T:stra}
		Suppose $0<\beta<\frac{2H}{n}$ and $\rtheta > 0$ small enough. Then 
		as $t\rightarrow 0$,
		\begin{align}
		\sigma_{\rm impl}(k_t,t)^2&=\sum_{j=0}^{n-2}\frac{(-1)^j2^j}{I^{\prime\prime}(0)^{j+1}}
		\left(\sum_{i=3}^n\frac{I^{(i)}(0)}{i!}
		k^{i-2}t^{(i-2)\beta}\right)^j \nonumber \\
		&\quad+\mathcal{O}\left(\phi_{n,H,\beta,\rtheta}(t)\right).
		\label{E:dada}
		\end{align}
		The $\mathcal{O}$-estimate in (\ref{E:dada}) depends on $n$,
                $H$, $\beta$, $\rtheta$, and $k$. It is uniform on compact subsets of $[0,\infty)$ with respect to the variable $k$.
	\end{theorem}
	\begin{remark}\label{R:po} \rm
		Using the multinomial formula, we can represent the expression on the left-hand side of (\ref{E:dada}) in terms of certain powers of $t$. However, the coefficients become rather complicated.
	\end{remark}
	\begin{remark}\label{R:pom} \rm Let an integer $n\ge 2$ be fixed, and suppose we would like to use only the derivatives $I^{(i)}(0)$ for $2\le i\le n$ in formula (\ref{E:dada}) to approximate $\sigma_{\rm impl}(k_t,t)^2$. Then, the optimal range for $\beta$ is the following: 
		$\frac{2H}{n+1}\le\beta<\frac{2H}{n}$. On the other hand, if $\beta$ is outside of the interval $[\frac{2H}{n+1},\frac{2H}{n})$, more derivatives of the energy function at zero may be needed to get a good approximation of the implied variance in formula (\ref{E:dada}).
	\end{remark}
	We will next derive from Theorem \ref{T:stra} several asymptotic formulas for the implied volatility. In the next corollary, we take $n=2$.
	\begin{corollary}\label{C:asymptot}
		As $t\rightarrow 0$, 
		\begin{align}
		\sigma_{\rm impl}(k_t,t)&=\sigma_0
		+\mathcal{O}(\phi_{2,H,\beta,\rtheta}(t)).
		\label{E:ost}
		\end{align}
	\end{corollary}
	
	Corollary \ref{C:asymptot} follows from Theorem \ref{T:stra} with $n=2$, the equality 
	\begin{equation}
	I^{\prime\prime}(0)=\sigma_0^{-2}
	\label{E:ee1}
	\end{equation} 
	given in Theorem \ref{thm:mod}, and the Taylor expansion $\sqrt{1+h}=1+\mathcal{O}(h)$ as $h\rightarrow 0$. 
	
	In the next corollary, we consider the case where $n=3$.
	\begin{corollary}\label{C:asympt}
		Suppose $\beta<\frac{2H}{3}$. Then,
		as $t\rightarrow 0$, 
		\begin{align}
		\sigma_{\rm impl}(k_t,t)&=\sigma_0
		+\rho\frac{\sigma_0^{\prime}}{\sigma_0}\langle K1,1\rangle kt^{\beta}+\mathcal{O}(\phi_{3,H,\beta,\rtheta}(t)).
		\label{E:o}
		\end{align}
	\end{corollary}
	
	Corollary \ref{C:asympt} follows from Theorem \ref{T:stra} with $n=3$, formula (\ref{E:ee1}), the equality
	\begin{equation}
	I^{\prime\prime\prime}(0)=-6\rho\frac{\sigma_0^{\prime}}{\sigma_0^4}\langle K1,1\rangle
	\label{E:ee2}
	\end{equation}
	(see Theorem \ref{thm:mod}), and the expansion $\sqrt{1+h}=1+\frac{1}{2}h+\mathcal{O}(h^2)$ as $h\rightarrow 0$.
	
	Using Corollary \ref{C:asympt}, we establish the following implied volatility skew formula in the moderate deviation regime.
	
	\begin{corollary}
		\label{C:skew}
	Let $0< H<\frac{1}{2}$, $0<\beta<\frac{2}{3}H$, and fix $y,z> 0$ with $y\neq z$. Then as
	$t\rightarrow 0$,
	\begin{equation}
	\frac{\sigma_{\rm impl}(yt^{\frac{1}{2}-H+\beta},t)-\sigma_{\rm impl}(zt^{\frac{1}{2}-H+\beta},t)}
	{(y-z)t^{\frac{1}{2}-H+\beta}}\sim\rho\frac{\sigma_0^{\prime}}{\sigma_0}\langle K1,1\rangle t^{H-\frac{1}{2}}.
	\label{E:te}
	\end{equation}
	\end{corollary}
	
	\begin{remark}
		Corollary \ref{C:skew} complements earlier works of Al\`{o}s et al. \cite{ALV07} and Fukasawa \cite{Fuk11, Fuk17}. 
		For instance, the following formula can be found in \cite[p. 6]{Fuk17}, see also \cite[p. 14]{Fuk11}: 
		\begin{equation}
		\frac{\sigma_{\rm impl}(yt^{\frac{1}{2}},t)-\sigma_{\rm impl}(zt^{\frac{1}{2}},t)}
		{(y-z)t^{\frac{1}{2}}}\sim\rho C(H) \frac{\sigma_0^{\prime}}{\sigma_0}  t^{H-\frac{1}{2}}.
		\label{E:fukskew}
		\end{equation}
		In formula \eqref{E:fukskew}, we employ the notation used in the present paper. Our analysis shows that the applicability range
		of skew approximation formulas is by no means restricted to the Central Limit Theorem type log-moneyness deviations of order $t^{1/2}$.
		It also includes the moderate deviations regime of order $t^{1/2 - H + \beta}$. The previous rate is clearly $\gg t^{1/2}$ as $t \to 0$.
		
	\end{remark}
	


\begin{remark}[Symmetry]
Write $\Phi_{1}(W,B,\widehat{B};\rho;\sigma)$ for the ``It\^{o}-type map''
\begin{equation*}
  \Phi _{1}(W,B,\widehat{B}) \coloneqq \int_{0}^{1}\sigma (\widehat{B})d\left( \bar{\rho}W+\rho
    B\right) .  
\end{equation*}
 It equals, in law, $\Phi _{1}(W,-B,-\widehat{B};-\rho;\sigma (- \cdot))$, and
indeed all our formulae are invariant under this transformation. In particular, the skew remains unchanged when the pair $(\rho, \sigma_0')$ is replaced by $(-\rho,-\sigma_0')$.
\end{remark}

\section{Simulation results}\label{sec:SimulationResults}

We verify our theoretical results numerically with a variant of the \emph{rough Bergomi model} \cite{BFG16} which fits nicely into the general rough volatility framework considered in this paper. As before, the model has been normalized such that $S_0 = 1$ and $r=0$. We let $(W,B)$ be two independent Brownian motions and $\rho \in (-1,1)$ with $\bar{\rho}^2 = 1 - \rho^2$ such that $Z = \bar{\rho} W +  \rho B$ is another Brownian motion having constant correlation $\rho$ with $B$. For some spot volatility $\sigma_0$ and volatility of volatility parameter $\eta$, we then assume the following dynamics for some asset $S$:

\begin{align}
\frac{dS_t}{S_t} &= \sigma(\hat{B}_t) dZ_t \label{eq:mrb_asset_sde}\\
\sigma(x) &= \sigma_0 \exp\left(\frac{1}{2} \eta x\right) \label{eq:mrb_vol_expl}
\end{align}	
where $\hat{B}$ is a Riemann-Liouville fBM given by
$$\widehat{B}_{t} =  \sqrt{2H} \int_{0}^{t} |t-s|^{H-1/2} dB_s.$$

The approach taken for the Monte Carlo simulations of the quantities we are interested in is the one initially explored in the original \emph{rough Bergomi} pricing paper \cite{BFG16}. That is, exploiting their joint Gaussianity, where we use the well-known Cholesky method to simulate the joint paths of $(Z, \hat{B})$ on some discretization grid $\mathcal{D}$. With \eqref{eq:mrb_vol_expl} being an explicit function in terms of the rough driver, an Euler discretisation of the Ito SDE \eqref{eq:mrb_asset_sde} on $\mathcal{D}$ then yields estimates for the price paths.

The Cholesky algorithm critically hinges on the availability and explicit computability of the joint covariance matrix of $(Z, \hat{B})$ whose terms we readily compute below.\footnote{ Note that expressions for the exact same scenario have have been computed before in the original pricing paper \cite{BFG16}, yet in that version the expression for the autocorrelation of the fBM $\hat{B}$ was incorrect. We compute and state here all the relevant terms for the sake of completeness.}
\begin{lemma}
	\label{lemma: cov_BM_fBM}
	For convenience, define constants $\gamma = \frac{1}{2} - H \in [0,\frac{1}{2})$ and $D_H = \frac{\sqrt{2H}}{H+\frac{1}{2}}$ and define an auxiliary function $G: [1,\infty) \rightarrow \mathbb{R}$ by
	\begin{align}
	G(x) = 2H \left(\frac{1}{1-\gamma}x^{-\gamma} +\frac{\gamma}{1-\gamma} x^{-(1+\gamma)}\frac{1}{2-\gamma}{}_2F_1(1,1+\gamma,3-\gamma, x^{-1}) \right)
	\end{align}
	where $_2 F_1$ denotes the Gaussian hypergeometric function \cite{Olv10}. Then the joint process $(Z, \hat{B})$ has zero mean and covariance structure governed by \\
	
	$
	\begin{cases}
	\Var[\hat{B}_t^2] = t^{2H}, & \text{for $t \geq 0$,}\\
	\Cov[\hat{B}_s \hat{B}_t] = t^{2H}G\left(s/t\right), & \text{for $s > t \geq 0$,}\\
	\Cov[\hat{B}_s Z_t] = \rho D_H \left(s^{H + \frac{1}{2}} - \left(s-\min(t,s)\right)^{H + \frac{1}{2}}\right), & \text{for $t, s \geq 0$,}\\
	\Cov[Z_tZ_s] =\min(t,s), & \text{for $t, s \geq 0$.}
	\end{cases}
	$	
\end{lemma}

Numerical simulations\footnote{ The Python 3 code used to run the simulations can be found at \href{https://github.com/RoughStochVol}{github.com/RoughStochVol}.} confirm the theoretical results obtained in the last section.
In particular - as can be seen in Figure~\ref{fig:pub_roughimpvol} -- the
asymptotic formula for the implied volatility \eqref{E:o} captures very well
the geometry of the term structure of implied volatility, with particularly good results for higher $H$ and worsening results as $H \downarrow
0$. Quite surprisingly, despite being an asymptotic formula, it seems to be
fairly accurate over a wide array of maturities extending up to a single
year.%
\afterpage{
	\begin{landscape}
	\begin{figure}
	\centering
	\captionsetup{width=\linewidth}
	\includegraphics[width=0.75\textwidth]{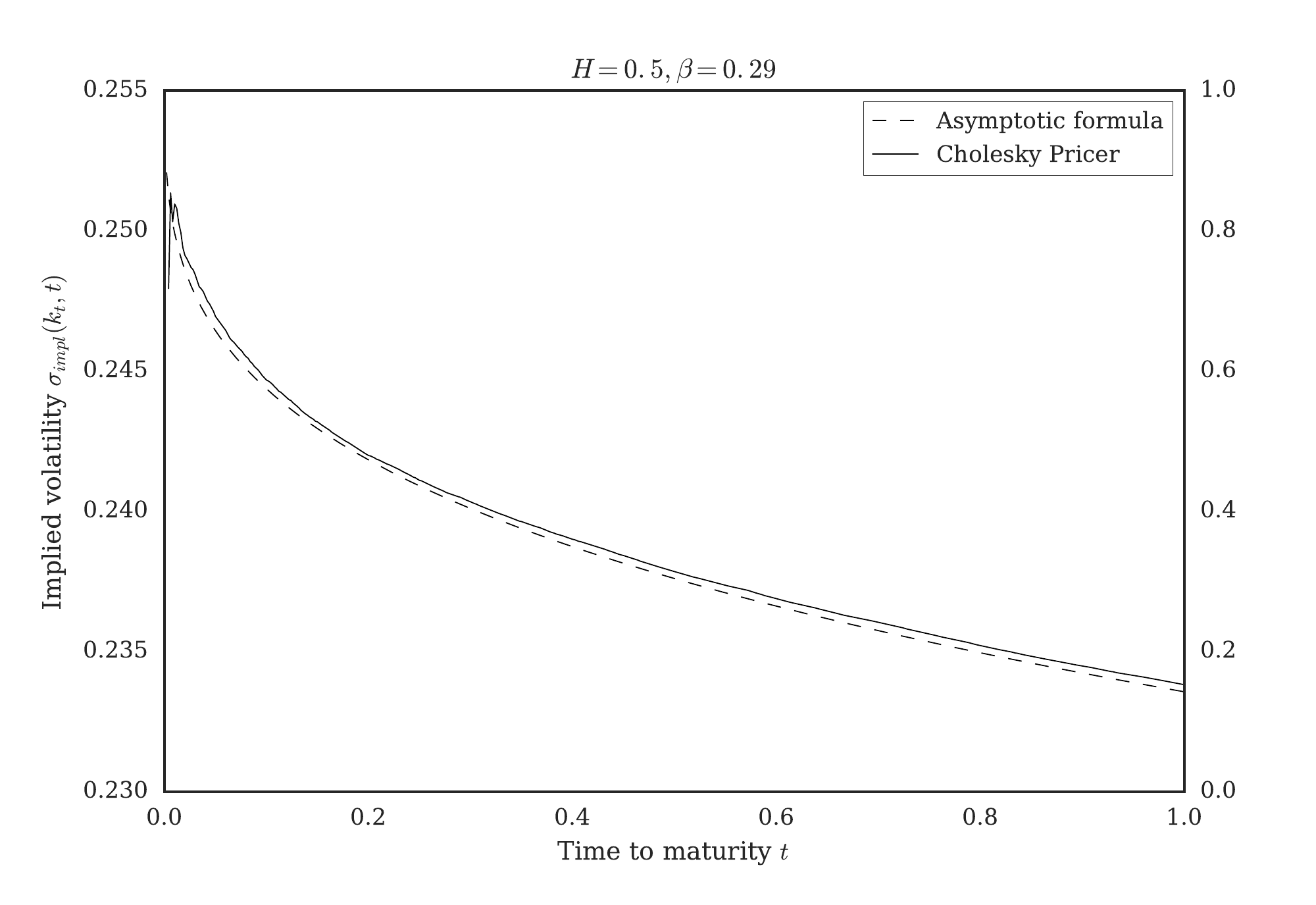}
	\includegraphics[width=0.75\textwidth]{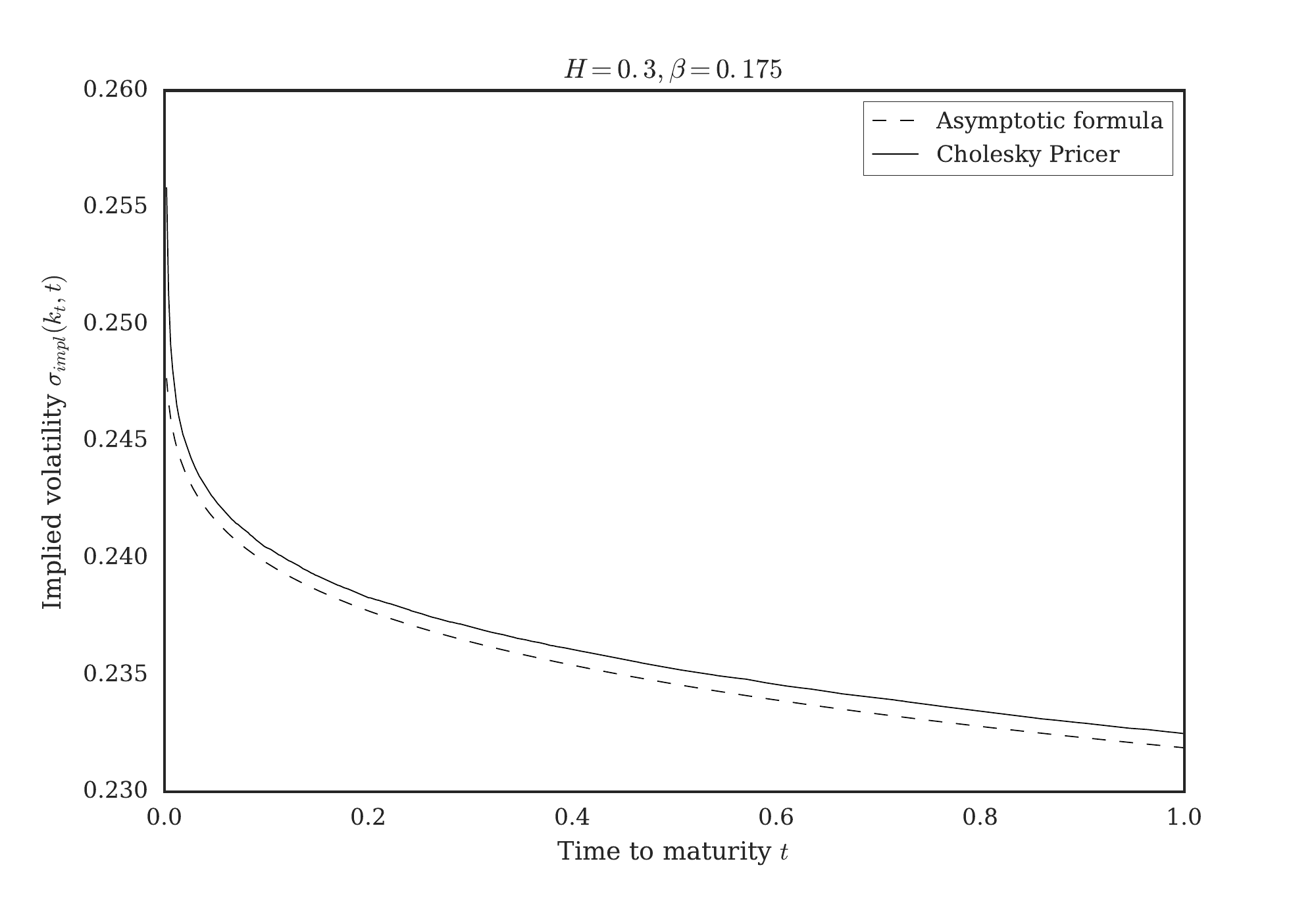}
	\includegraphics[width=0.75\textwidth]{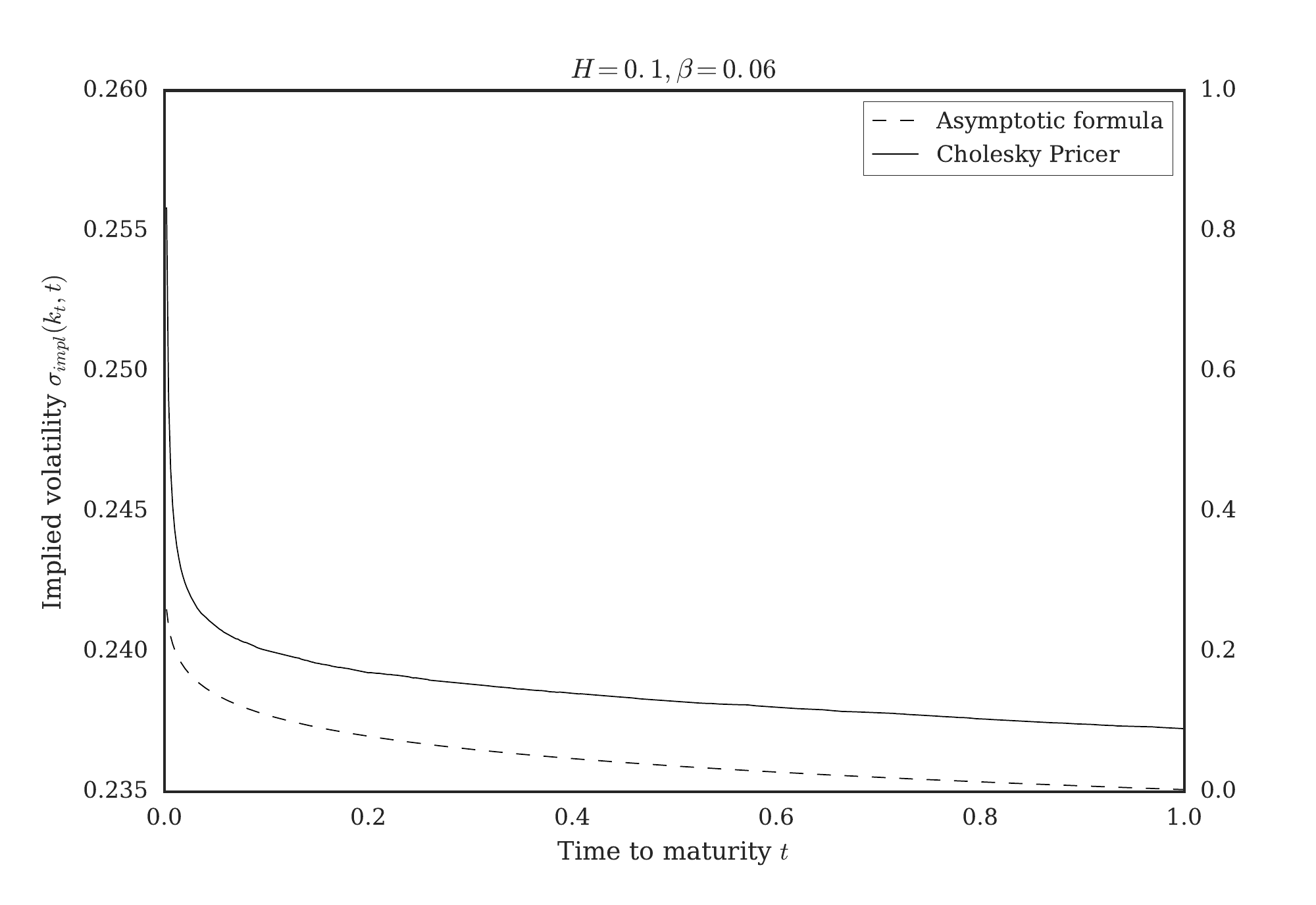}
	\caption{Illustration of the term structure of implied volatility of the Modified Rough Bergomi model in the Moderate deviations regime with time-varying log-strike $k_t = 0.4t^{\beta}$. Depicted are the asymptotic formula (Eq. \eqref{E:o}, dashed line) and an estimate based on $N=10^8$ samples of a MC Cholesky Option Pricer (solid line) with $500$ time steps. Model parameters are given by spot vol $\sigma_0 \approx 0.2557$, vvol $\eta= 0.2928$ and correlation parameter $\rho=-0.7571$.}
	\label{fig:pub_roughimpvol}
	\end{figure}

\end{landscape}
}

%
\section{Proof of the energy expansion}
\label{sec:energy-expansion}

Consider
\begin{eqnarray*}
dX &=& -\half \sigma^2(Y) dt + \sigma \left( Y\right) d\left( \bar{\rho}dW+\rho dB\right)
,\,\,\,X_{0}=0 \\
dY &=&d\widehat{B},\,\,Y_{0}=0
\end{eqnarray*}%
where $\widehat{B}_{t}=\int_{0}^{t}K\left( t,s\right) dB_{s}$ for a fixed
Volterra kernel (recall \eqref{eq:LogpriceDrift} in the previous section).
We study the small noise problem $\left(
X^{\varepsilon },Y^{\varepsilon }\right) $ where $\left( W,B,\widehat{B}\right) $
is replaced by $\left( \varepsilon W,\varepsilon B, \hatepsilon \widehat{B}%
\right) $. 
The following proposition roughly says that%
\[
\mathbb{P}\left( X_{1}^{\varepsilon }\approx \frac{\varepsilon}{\hatepsilon} x\right)
\approx \exp \left( -\frac{I\left( x\right) }{\hatepsilon^2}\right) .
\]

\begin{proposition}[Forde-Zhang \cite{FZ17}]
\label{prop:fractional-LDP}
Under suitable assumptions (cf.~Section~\ref{sec:ExpositionAssumptions}), the rescaled process $\left( \f{\hatepsilon}{\varepsilon} X_{1}^{\varepsilon }:\varepsilon \geq 0\right) $
satisfies an LDP\ (with speed $\hatepsilon^{2}$) and rate function%
\begin{equation}
I\left( x\right) =\inf_{f\in H_0^1}\left[ \frac{\left( x-\rho G\left(
f\right) \right) ^{2}}{2\bar{\rho}^{2}F\left( f\right) }+\frac{1}{2}E\left( f
\right) \right] 
\equiv \inf_{f\in H_0^1} \mathcal{I}_x\left( f\right) 
\label{eq:minimization-problem}
\end{equation}
where
\begin{eqnarray*}
G\left( f\right)  &=&\int_{0}^{1}\sigma \left( \left( K\dot{f}\right) \left(
s\right) \right) \dot{f}_{s}ds\equiv \left\langle \sigma \left( K\dot{f}%
\right) ,\dot{f}\right\rangle \equiv \left\langle \sigma (\widehat{f}),\dot{f}%
\right\rangle  \\
F\left( f\right)  &=&\int_{0}^{1}\sigma \left( \left( K\dot{f}\right) \left(
s\right) \right) ^{2}ds\equiv \left\langle \sigma ^{2}\left( K\dot{f}\right)
,1\right\rangle \equiv \left\langle \sigma ^{2}(\widehat{f}),1\right\rangle  \\
E\left( f\right)  &=&\int_{0}^{1}\left\vert \dot{f}\left( s\right)
\right\vert ^{2}ds\equiv \left\langle \dot{f},\dot{f}\right\rangle 
\end{eqnarray*}
\end{proposition}

The rest of this section is devoted to analysis of the function $I$ as defined
in~\eqref{eq:minimization-problem}. 
First, we derive the first order optimality condition for the above minimization problem.
\begin{proposition}[First order optimality condition]
\label{prop:first-order-optimality}
For any $x\in \R$ we have at any local minimizer $f=f^{x}$ of the functional $\mathcal{I}_x$ in \eqref{eq:minimization-problem} 
that
\begin{multline}\label{eq:frelation}
f_{t}^{x} =\frac{\rho \left( x-\rho G\left( f^{x}\right) \right) \left\{
\left\langle \sigma \left( K\dot{f}^{x}\right) ,1_{\left[ 0,t\right]
}\right\rangle +\left\langle \sigma ^{\prime }\left( K\dot{f}^{x}\right) 
\dot{f}^{x},K1_{\left[ 0,t\right] }\right\rangle \right\} }{\bar{\rho}%
^{2}F\left( f^{x}\right) } \\
+\frac{\left( x-\rho G\left( f^{x}\right) \right) ^{2}}{\bar{\rho}%
^{2}F^{2}\left( f^{x}\right) }\left\langle \left( \sigma \sigma ^{\prime
}\right) \left( K\dot{f}^{x}\right) ,K1_{\left[ 0,t\right] }\right\rangle,
\end{multline}
for all $t\in \left[ 0,1\right]$.
\end{proposition}

\begin{proof}
We denote $a\approx b$ whenever $a=b+o\left( \delta \right) $ for a small
parameter $\delta$. We expand
\begin{eqnarray*}
E\left( f+\delta g\right)  &\approx &E\left( f\right) +2\delta \left\langle 
\dot{f},\dot{g}\right\rangle  \\
F\left( f+\delta g\right)  &\approx &F\left( f\right) +\delta \ip{
\left( \sigma ^{2}\right) ^{\prime }\left( K\dot{f}\right)}{K\dot{g}} \\
G\left( f+\delta g\right)  &\approx &G\left( f\right) +\delta \left\{
\left\langle \sigma \left( K\dot{f}\right) ,\dot{g}\right\rangle
+\left\langle \sigma ^{\prime }\left( K\dot{f}\right) \dot{f},K\dot{g}%
\right\rangle \right\} 
\end{eqnarray*}%
If $f=f^{x}$ is a minimizer then $\delta \mapsto \mathcal{I}_x\left( f+\delta g\right) $
has a minimum at $\delta =0$ for all $g$. We expand%
\begin{align*}
\mathcal{I}_x\left( f+\delta g\right) &= \frac{\left( x-\rho G\left( f+\delta g\right)
\right) ^{2}}{2\bar{\rho}^{2}F\left( f+\delta g\right) }+\frac{1}{2}E(f+\delta g)  \\
&\approx \frac{\left( x-\rho G\left( f\right) -\delta \rho \left\{
\left\langle \sigma \left( K\dot{f}\right) ,\dot{g}\right\rangle
+\left\langle \sigma ^{\prime }\left( K\dot{f}\right) \dot{f},K\dot{g}%
\right\rangle \right\} \right) ^{2}}{2\bar{\rho}^{2}\left[ F\left( f\right)
+\delta \left\langle \left( \sigma ^{2}\right) ^{\prime }\left( K\dot{f}%
\right) , K\dot{g}\right\rangle \right] }\\
&\quad +\frac{1}{2}E( f ) +\delta
\left\langle \dot{f},\dot{g}\right\rangle  \\
&\approx \frac{\left( x-\rho G\left( f\right) \right) ^{2}-\delta 2\rho
\left( x-\rho G\left( f\right) \right) \left\{ \left\langle \sigma \left( K%
\dot{f}\right) ,\dot{g}\right\rangle +\left\langle \sigma ^{\prime }\left( K%
\dot{f}\right) \dot{f},K\dot{g}\right\rangle \right\} }{2\bar{\rho}%
^{2}F\left( f\right) \left[ 1+\frac{\delta }{F\left( f\right) }\left\langle
\left( \sigma ^{2}\right) ^{\prime }\left( K\dot{f}\right) ,K\dot{g}%
\right\rangle \right] }\\
&\quad+\frac{1}{2}E\left( f\right) +\delta \left\langle 
\dot{f},\dot{g}\right\rangle  \\
&\approx \frac{\left( x-\rho G\left( f\right) \right) ^{2}-\delta 2\rho
\left( x-\rho G\left( f\right) \right) \left\{ \left\langle \sigma \left( K%
\dot{f}\right) ,\dot{g}\right\rangle +\left\langle \sigma ^{\prime }\left( K%
\dot{f}\right) \dot{f},K\dot{g}\right\rangle \right\} }{2\bar{\rho}%
^{2}F\left( f\right) } \\
&\quad-\frac{\left( x-\rho G\left( f\right) \right) ^{2}}{2\bar{\rho}^{2}F\left(
f\right) }\frac{\delta }{F\left( f\right) }\left\langle \left( \sigma
^{2}\right) ^{\prime }\left( K\dot{f}\right) , K\dot{g}\right\rangle +\frac{1}{2}E\left( f\right) +\delta \left\langle \dot{f},\dot{g}%
\right\rangle .
\end{align*}%
As a consequence, we must have, for $f=f^{x}$ and every $\dot{g}\in L^{2}%
\left[ 0,1\right] $ 
\begin{eqnarray*}
0 &=&\frac{d}{d\delta }\left\{ \mathcal{I}_x\left( f+\delta g\right) \right\} _{\delta
=0}=-\frac{\rho \left( x-\rho G\left( f\right) \right) \left\{ \left\langle
\sigma \left( K\dot{f}\right) ,\dot{g}\right\rangle +\left\langle \sigma
^{\prime }\left( K\dot{f}\right) \dot{f},K\dot{g}\right\rangle \right\} }{%
\bar{\rho}^{2}F\left( f\right) } \\
&&-\frac{\left( x-\rho G\left( f\right) \right) ^{2}}{\bar{\rho}%
^{2}F^{2}\left( f\right) }\left\langle \left( \sigma \sigma ^{\prime
}\right) \left( K\dot{f}\right) , K\dot{g}\right\rangle +\left\langle
   \dot{f},\dot{g}\right\rangle.
\end{eqnarray*}%

Recall $f_{0}^{x}=0$, any $x$.
We now test with $\dot{g}=1_{\left[ 0,t\right] }$ for a fixed $t\in \left[
0,1\right] $ and obtain%
\begin{align*}
f_{t}^{x} &=\frac{\rho \left( x-\rho G\left( f^{x}\right) \right) \left\{
\left\langle \sigma \left( K\dot{f}^{x}\right) ,1_{\left[ 0,t\right]
}\right\rangle +\left\langle \sigma ^{\prime }\left( K\dot{f}^{x}\right) 
\dot{f}^{x},K\indic{\left[ 0,t\right] }\right\rangle \right\} }{\bar{\rho}%
^{2}F\left( f^{x}\right) } \\
&\quad +\frac{\left( x-\rho G\left( f^{x}\right) \right) ^{2}}{\bar{\rho}%
^{2}F^{2}\left( f^{x}\right) }\left\langle \left( \sigma \sigma ^{\prime
}\right) \left( K\dot{f}^{x}\right) , K \indic{\left[ 0,t\right] }\right\rangle.
\qedhere
\end{align*}
\end{proof}

\subsection{Smoothness of the energy}
\label{sec:smoothness-energy}

Having formally identified the first order condition for minimality
in~\eqref{eq:minimization-problem}, we will now show that the energy
$x \mapsto I(x)$ is a smooth function. More precisely, we will use the
implicit function theorem to show that the minimizing configuration $f^x$ is a
smooth function in $x$ (locally at $x = 0$). As $\mathcal{I}_x$ is a smooth
function, too, this will imply smoothness of
$x \mapsto \mathcal{I}_x(f^x) = I(x)$, at least in a neighborhood of
$0$.

As the Cameron-Martin space $\mathcal{H}$ of the process $\hat{B}$ continuously
embeds into $C\left([0,1]\right)$, $K$ maps $H^1_0$ continuously into
$C\left([0,1]\right)$, i.e., there is a constant $C > 0$ such that for any $f
\in H^1_0$ we have
\begin{equation}
  \label{eq:Kf-infty-bound}
  \norm{K\dot{f}}_\infty \le C \norm{f}_{H^1_0}.
\end{equation}
This result will follow from
\begin{lemma}\label{lem:cameron-martin-space}
Let $\left( V_{t}:0\leq t\leq 1\right) $ be a continuous, centred Gaussian
process and $\mathcal{H}$ its Cameron-Martin space. Then we have the
continuous embedding $\mathcal{H}\hookrightarrow C\left[ 0,1\right] $. That
is, for some constant $C$, 
\[
\norm{h}_{\infty }\leq C \norm{h}_{\mathcal{H}}.
\]
\end{lemma}
\begin{proof}
By a fundamental result of Fernique, applied to the law of $V$ as Gaussian
measure on the Banach space $(C\left[ 0,1\right] ,\left\Vert \cdot
\right\Vert _{\infty })$, the random variable $\left\Vert V\right\Vert
_{\infty }$ has Gaussian integrability. In particular, 
\[
\sigma ^{2}:=\mathbb{E(}\left\Vert V\right\Vert _{\infty }^{2})<\infty \text{%
,}
\]%
On the other hand, a generic element $h\in \mathcal{H}$ can be written
  as $%
  h_{t}=E\left[ V_{t}Z\right] $ where $Z$ is a centred Gaussian random
  variable with variance $\left\Vert h\right\Vert _{\mathcal{H}}^{2}$,
  see, e.g.,~\cite[page 150]{friz2014course}. By Cauchy--Schwarz,%
\[
\left\vert h_{t}\right\vert \leq E\left[ \left\vert
V_{t}\right\vert \right] ^{1/2}\left\Vert h\right\Vert _{\mathcal{H}}\leq
\sigma \left\Vert h\right\Vert _{\mathcal{H}}
\]%
and conclude by taking the $\sup $ over on the l.h.s. over $t\in \left[ 0,1%
\right] $.
\end{proof}

\begin{remark}
  \label{rem:assumption-K}
  Assume $V$ is of Volterra form, i.e. $V_{t}=\int_{0}^{t}K\left( t,s\right)
  dB_{s}$. Then it can be shown (see \cite[Section 3]{Dec05})
  that $\mathcal{H}$ is the image of $L^{2}$ under the map%
  \[
  K:\dot{f}\mapsto \hat{f}:=\left( t\mapsto \int_{0}^{t}K\left( t,s\right) 
    \dot{f}_{s}ds\right) 
  \]%
  and $\left\Vert K\dot{f}\right\Vert _{\mathcal{H}}=\left\Vert \dot{f}%
  \right\Vert _{L^{2}}.$ In particular then, applying the above with $h=K\dot{f}%
  \in \mathcal{H}$, gives%
  \[
  \left\Vert K\dot{f}\right\Vert _{\infty }\leq C \left\Vert K\dot{f}%
  \right\Vert _{\mathcal{H}}= C \left\Vert \dot{f}\right\Vert _{L^{2}} = C \norm{f}_{H^1_0}.
  \]
\end{remark}

\subsubsection{The uncorrelated case}
\label{sec:uncorrelated-case}

We start with the case $\rho = 0$ as the formulas are much simpler in this case.

By Proposition~\ref{prop:first-order-optimality}, any local optimizer
$f = f^x$ of the functional $\mathcal{I}_x: H_0^1 \to
\R$ 
in the uncorrelated case $\rho = 0$ satisfies for any
$t \in [0,1]$
\begin{equation*}
f_t = \frac{x^{2}}{F^{2}\left( f \right) }\left\langle \left( \sigma \sigma
    ^{\prime }\right) \left( K\dot{f}\right) ,K\indic{[ 0,t]} \right\rangle.
\end{equation*}
We define a map $H: H_0^1 \times \R \to H_0^1$ by
\begin{equation}
  \label{eq:implicit-function-def}
  H(f,x)(t) \coloneqq f_t - \frac{ x^2 }{ F^{2}\left( f \right) } \ip{\left(
      \sigma \sigma ^{\prime} \right) \left( K\dot{f}\right) }{ K\indic{[0,t]}}.
\end{equation}
Hence, for given $x \in \R$, any local optimizer $f$ must solve $H(f,x) =
0$. As one particular solution is given by the pair $(0,0)$, we are in the
realm of the implicit function theorem. We need to prove that
\begin{itemize}
\item $(f,x) \mapsto H(f,x)$ is locally smooth (in the sense of Fr\'{e}chet);
\item $DH(f,x) \coloneqq \f{\pa}{\pa f}H(f,x)$ is invertible in $(0,0)$.
\end{itemize}
Note that invertibility should hold for $x$ small enough, as
$DH(f,x) = \id_{H_0^1} - x^2 R$ for some $R$, which is invertible as long as
$R$ has a bounded norm for sufficiently small $x$.

\begin{remark}
  The method of proof in this section is purely local in
  $H^1_0$. Hence, we only really need smoothness of $\sigma$
    locally around $0$. Note, however, that stochastic Taylor expansions
  used in Section~\ref{sec:price-expansion} will actually require global
  smoothness of $\sigma$.
\end{remark}

\begin{lemma}
  \label{lem:smoothness-F-R1}
  The functions $F:H_0^1 \to \R$ and $R_1: H_0^1 \to C\left([0,1]\right)$
  defined by
  \begin{equation*}
    R_1(f)(t) \coloneqq \ip{\left( \sigma \sigma ^{\prime} \right) \left(
        K\dot{f}\right) }{ K\indic{[0,t]}}, \quad t \in [0,1],
  \end{equation*}
  are smooth in the sense of Fr\'{e}chet.
\end{lemma}
\begin{proof}
  For $N \ge 1$ we note that the Gateaux derivative of $F$ satisfies
  \begin{equation*}
    D^NF(f) \cdot \left( g_1, \ldots, g_N \right) = \int_0^1 \f{d^N}{dx^N}
    \sigma^2( K \dot{f} ) K\dot{g}_1 \cdots K \dot{g}_N ds.
  \end{equation*}
  By Lemma~\ref{lem:cameron-martin-space}, we can bound
  \begin{align*}
    \abs{D^NF(f) \cdot \left(g_1, \ldots, g_N \right)} 
    &\le \const \int_0^1 \abs{K\dot{g_1}(s)} \cdots \abs{K\dot{g_N}(s)} ds\\
    &\le \const \ \norm{K\dot{g}_1}_\infty \cdots \norm{K\dot{g}_N}_\infty\\
    &\le \const \ C^{N} \norm{g_1}_{H^1_0} \cdots \norm{g_N}_{H^1_0},
  \end{align*}
  for $\const = \norm{\f{d^n}{dx^n} \sigma^2}_\infty$.\footnote{More
    precisely, since neither $\sigma$ nor its derivatives need to be bounded,
    we need to actually work with a local version of the above estimate, for
    instance by replacing the max with a sup over a compact set containing
    $\{(K\dot{f})(t) : 0 \le t \le 1\}$.}
  Thus, $D^NF(f)$ is a multi-linear form on $H^1_0$ with operator norm
  $\norm{D^NF(f)} \le \norm{\f{d^n}{dx^n} \sigma^2}_\infty
  C^{N}$ independent of $f$. As $f \mapsto D^NF(f)$ is
  continuous, we conclude that $D^NF(f)$ as given above is, in fact, a
  Fr\'{e}chet derivative.

  Let us next consider the functional $R_1$. Note that
  \begin{equation*}
    \left( D^NR_1(f) \cdot (g_1, \ldots, g_N) \right)(t) = \ip{\fs_N(K\dot{f})
      K\dot{g}_1 \cdots K\dot{g}_N}{K\indic{[0,t]}}
  \end{equation*}
  for $\fs_N(x) \coloneqq \f{d^N}{dx^N} \sigma(x) \sigma^\prime(x)$. Hence,
  Assumption~\ref{ass:conditions} implies that
  \begin{align*}
    \norm{D^NR_1(f) \cdot (g_1, \ldots, g_N)}_{H^1_0}^2 &= \int_0^1 \left(
      \int_t^1 \fs_N\left( (K\dot{f})(s) \right)
      \prod_{i=1}^N (K\dot{g}_i)(s) K(s,t) ds \right)^2 dt\\
    &\le \norm{\fs_N}_\infty^2 \prod_{i=1}^N \norm{K\dot{g}_i}_\infty^2
    \int_0^1 \int_t^1 K(s,t)^2 ds dt\\
    &\le \norm{\fs_N}_\infty^2 \ C^{2N} \prod_{i=1}^N
    \norm{g_i}_{H_0^1}^2 \int_0^1 \int_0^s K(s,t)^2 dt ds\\
    &\le \norm{\fs_N}_\infty^2 \ C^{2N} \int_0^1 \int_0^s
      K(s,t)^2 dt ds \prod_{i=1}^N \norm{g_i}_{H_0^1}^2.
  \end{align*}
  We see that the multi-linear map $D^NR_1(f)$ has operator norm bounded by
  \begin{equation*}
    \norm{D^NR_1(f)} \le \norm{\fs_N}_\infty \ C^N
    \sqrt{\int_0^1 \int_0^s K(s,t)^2 dt ds},
  \end{equation*}
  independent of $f$. From continuity of $f \mapsto D^NR_1(f)$, it follows
  that $D^NR_1(f)$ is the $N$'th Fr\'{e}chet derivative.
\end{proof}

\begin{theorem}[Zero correlation]
  \label{thr:smoothness-energy-rho=0}
  Assuming $\rho=0$, the energy $I(x)$ (as defined in~\eqref{eq:minimization-problem}) is
  smooth in a neighborhood of $x = 0$.
\end{theorem}
\begin{proof}
  By construction, we have
  \begin{equation*}
    DH(f,x) = \id_{H_0^1} - x^2 A(f)
  \end{equation*}
  for $A: H_0^1 \to \mathcal{L}(H_0^1, H_0^1)$ defined by
  \begin{equation*}
    A(f) \coloneqq R_1(f) \otimes DF^{-2}(f) + F^{-2}(f) DR_1(f).
  \end{equation*}
  Here,
  \begin{equation*}
    \left( R_1(f) \otimes DF^{-2}(f) \right) \cdot g = \underbrace{(DF^{-2}(f)
      \cdot g)}_{\in \R} \underbrace{R_1(f)}_{\in H_0^1}.
  \end{equation*}
  As verified above, $H$ is smooth in the sense of Fr\'{e}chet. Trivially,
  $DH(0,0) = \id_{H^1_0}$ is invertible and $H(0,0) = 0$. Therefore, the
  implicit function theorem implies that there are open neighborhoods $U$ and
  $V$ of $0 \in H^1_0$ and $0 \in \R$, respectively, and a smooth map $x
  \mapsto f^x$ from $V$ to $U$ such that $H(f^x,x) \equiv 0$ and $f^x$ is
  unique in $U$ with this property.
  
  For the energy, we prove that $I(x) = \mathcal{I}_x(f^x)$ in a neighborhood of $x =
  0$. First of all, we show that a minimizer exists. If not, there is a
  function $g \in H^1_0$ with $\mathcal{I}_x(g) < \mathcal{I}_x(f^x)$. For small
  enough $x$ such a $g$ must be inside a ball with radius $\epsilon$ around $0
  \in H^1_0$, as $\mathcal{I}_x(g) \ge \half \norm{g}^2_{H_0^1}$ and $\lim_{x \to 0}
  \mathcal{I}_x(f^x) = 0$. Then note that for any $g \in H^1_0$
  \begin{equation*}
    D^2\mathcal{I}_0(0) \cdot (g,g) = \norm{g}_{H_0^1}^2 > 0,
  \end{equation*}
  where $D^2\mathcal{I}_x(f)$ denotes the second derivative of $f \mapsto \mathcal{I}_x(f)$.
  By continuity, $D^2\mathcal{I}_x(f)$ stays positive definite for $(x,f)$ in a
  neighborhood of $(0,0)$. As noted, for $x$ small enough, both $g$ and $f^x$
  (and the line connecting them) lie in this neighborhood. For $h \coloneqq g
  - f^x$, this implies
  \begin{equation*}
    \mathcal{I}_x(g) - \mathcal{I}_x(f^x) = D\mathcal{I}_x(f_x) \cdot h + \int_0^1 D^2\mathcal{I}_x(f^x + t h) \cdot
    (h,h) \ dt > 0,
  \end{equation*}
  since $D\mathcal{I}_x(f_x) \cdot h = 0$ and $D^2\mathcal{I}_x(f^x + ts h) \cdot (h,h) > 0$. This
  contradicts the assumption that $\mathcal{I}_x(g) < \mathcal{I}_x(f^x)$, and we conclude that
  $f^x$ is, indeed, a minimizer of $\mathcal{I}_x$, implying that $I(x) =
  \mathcal{I}_x(f^x)$ locally.

  Finally, as $x \mapsto f^x$ is smooth and $(f,x) \mapsto \mathcal{I}_x(f) = \f{x^2}{2
    F(f)} + \half \norm{f}_{H^1_0}^2$ is smooth, we see that $x \mapsto I(x) =
  \mathcal{I}_x(f^x)$ is smooth in a neighborhood of $0$. (Note that this
  arguments relies on $\sigma(0) \neq 0$, implying that $F(f) \neq 0$ for $f$
  in a neighborhood to $0$.)
\end{proof}

\begin{remark}
  \label{rem:existence}
  Classical counter-examples in the context of the \emph{direct method} of
  calculus of variations show that the step of verifying the existence of a
  minimizer should not be taken too lightly. For instance, the functional
  \begin{equation*}
    J(u) \coloneqq \int_0^1 \left[ (u^\prime(s)^2 - 1)^2 + u(s)^2\right] ds
  \end{equation*}
  does not have a minimizer in $H^1_0$, but $J$ can be made arbitrarily close
  to $0$ by choosing piecewise-linear functions $u$ with slope $\abs{u^\prime}
  = 1$ oscillating around $0$. We refer to any text book on calculus of
  variations. In the situation above, local ``convexity'' in the sense of a
  positive definite second derivative prevents this phenomenon.  An
  alternative method of proof for the existence of a minimizer is to show that
  $J$ is (lower semi-) continuous in the weak sense.
\end{remark}

\subsubsection{The general case}
\label{sec:general-case}

In the general case (cf.~Proposition~\ref{prop:first-order-optimality}), we
define the function $H : H_0^1 \times \R \to H^1_0$ by
\begin{align}
  H(f,x)(t)
  &\coloneqq f_{t} - \frac{\rho \left( x-\rho G\left( f \right) \right) \left\{
      \ip{ \sigma \left( K\dot{f} \right)}{\indic{[0,t]}} + \ip{\sigma
        ^{\prime} \left( K\dot{f}\right) \dot{f}}{K\indic{[ 0,t]} } \right\}
  }{ \barrho^2 F\left( f \right) } \nonumber\\
  &\quad + \frac{\left( x-\rho G\left( f \right) \right)^2}{\barrho^2
    F^{2}\left( f \right) } \ip{\left( \sigma \sigma ^{\prime
      }\right) \left( K\dot{f}\right)}{ K\indic{[ 0,t]} } \nonumber\\
  &= f_t - \f{\rho\left( x - \rho G(f) \right)}{\barrho^2 F(f)}\left(
    R_2(f)(t) + R_3(f)(t) \right) + \f{\left( x - \rho G(f)
    \right)^2}{\barrho^2 F(f)^2} R_1(f)(t), \label{eq:H-correlated}
\end{align}
where $R_2, R_3: H_0^1 \to H^1_0$ are defined by
\begin{gather}
  \label{eq:R_2}
  R_2(f)(t) \coloneqq \ip{\sigma(K\dot{f})}{\indic{[0,t]}},\\
  \label{eq:R_3}
  R_3(f)(t) \coloneqq \ip{\sigma^\prime(K\dot{f}) \dot{f}}{K\indic{[0,t]}},
\end{gather}
$t \in [0,1]$. 

One easily checks that $G$, $R_2$, $R_3$ are smooth in the Fr\'{e}chet sense.

\begin{lemma}
  \label{lem:smoothness-G-R2-R3}
  The functions $G: H^1_0 \to \R$, $R_2: H^1_0 \to H^1_0$ and $R_3: H^1_0 \to
  H^1_0$ are smooth in Fr\'{e}chet sense.
\end{lemma}
\begin{proof}
  The proof of smoothness is clear. We report the actual derivatives. For $G$
  we get
   \begin{multline*}
   D^NG(f) \cdot \left( g_1, \ldots, g_N \right) = \ip{\sigma^{(N)}\left(
       K\dot{f} \right) \dot{f}}{ \prod_{i=1}^N K\dot{g}_i } +\\
   + \sum_{k=1}^N \ip{\sigma^{(N-1)}\left( K\dot{f} \right)}{\dot{g}_k
     \prod_{i \neq k} K\dot{g}_i}.
 \end{multline*}
 For $R_2$ and, respectively, $R_3$, we obtain
 \begin{equation*}
   \left(D^NR_2(f) \cdot (g_1, \ldots, g_N) \right)(t) = \int_0^t
   \sigma^{(N)}\left( (K\dot{f})(s) \right) \prod_{i=1}^N (K\dot{g}_i)(s) ds, 
 \end{equation*}
 and
 \begin{multline*}
   \left( D^NR_3(f) \cdot (g_1, \ldots, g_N) \right)(t) =
   \ip{\sigma^{(N+1)}\left( K\dot{f} \right) \dot{f} K\indic{[0,t]}}{
     \prod_{i=1}^N K \dot{g}_i} +\\
   + \sum_{k=1}^N \ip{\sigma^{(N)}\left(
       K\dot{f} \right) K\indic{[0,t]}}{\dot{g}_k \prod_{i \neq k} K\dot{g}_i}.\qedhere
 \end{multline*}
\end{proof}

\begin{theorem}
  \label{thr:smoothness-energy-general}
  Let $\sigma$ be smooth with $\sigma(0) \neq 0$. Then the energy $I(x)$ as
  defined in~\eqref{eq:minimization-problem} is smooth in a neighborhood of $x
  = 0$.
\end{theorem}
\begin{proof}
  The proof is similar to the proof of
  Theorem~\ref{thr:smoothness-energy-rho=0}. In fact, the only difference is
  in establishing invertibility of $DH(0,0)$ and the existence of a minimizer. 

  Note that~\eqref{eq:H-correlated} contains three terms. The derivative of
  the first term ($f \mapsto f$) is always equal to $\id_{H_0^1}$. For the
  second term, we note that
  \begin{equation*}
    \left.\left( x - \rho G(f) \right)\right|_{x=0,\, f=0} = 0.
  \end{equation*}
  Hence, the only non-vanishing contribution to the derivative of the second
  term evaluated in direction $g \in H_0^1$ at $x=0$, $f = 0$ and $t \in
  [0,1]$ is
  \begin{equation*}
    \f{\rho^2 DG(0) \cdot g}{\barrho^2 F(0)} (R_2(0) + R_3(0)) = \f{\rho^2
      \sigma_0 g(1)}{\barrho^2 \sigma_0^2}\left( \sigma_0 t + 0 \right) =
    \f{\rho^2}{\barrho^2} g(1) t. 
  \end{equation*}
  For the same reason, the derivative of the third term at $(f,x) = (0,0)$
  vanishes entirely. Hence,
  \begin{equation*}
    (DH(0,0) \cdot g)(t) = g(t) + \f{\rho^2}{\barrho^2} g(1) t.
  \end{equation*}
  It is easy to see that $g \mapsto DH(0,0) \cdot g$ is invertible. Indeed,
  let us construct the pre-image $g = DH(0,0)^{-1} \cdot h$ of some $h \in
  H^1_0$. At $t = 1$ we have
  \begin{equation*}
    \f{\barrho^2 + \rho^2}{\barrho^2}  g(1) = h(1),
  \end{equation*}
  implying $g(1) = \barrho^2 h(1)$. For $0 \le t < 1$, we then get
  \begin{equation*}
    g(t) + \f{\rho^2}{\barrho^2} g(1) t = g(t) + \f{\rho^2}{\barrho^2}
    \barrho^2 h(1) t = g(t) + \rho^2 h(1) t = h(t),
  \end{equation*}
  or $g(t) = h(t) - \rho^2 h(1) t$.

  For existence of the minimizer, note that 
  \begin{equation*}
    D^2J_0(0) \cdot (g,g) = \f{\rho^2}{\barrho^2} g(1)^2 + \norm{g}_{H_0^1}^2,
  \end{equation*}
  which is again positive definite.
\end{proof}

\begin{remark}
  Note that we do not really need infinite smoothness of $\sigma$ if we only
  want partial smoothness of $I$. Indeed, it is easy to show that $\sigma \in
  C^k$ implies that $I \in C^{k-1}$ (locally at $0$).
\end{remark}

\subsection{Energy expansion}
\label{sec:energy-expansion-1}

Having established smoothness of the energy $I$ as well as of the minimizing
configuration $x \mapsto f^x$ locally around $x = 0$, we can proceed with
computing the Taylor expansion of $f^x$ around $x = 0$. We will once more rely
on the first order optimality condition given in
Proposition~\ref{prop:first-order-optimality}. Plugging the Taylor expansion
of $f^x$ into $\mathcal{I}_x$ will then give us the local Taylor expansion of $I(x)$.

\subsubsection{Expansion of the minimizing configuration}
\label{sec:expans-minim-conf}

\begin{theorem}
  \label{thr:f^x-expansion-general-rho}
  We have 
  \begin{align*}
    f_{t}^{x} &=\alpha _{t}x+\beta _{t}\frac{x^{2}}{2}+\mathcal{O}\left( x^{3}\right), \\
    \alpha _{t} &= \frac{\rho }{\sigma _{0}}t, \\
    \beta _{t} &= 2 \f{\sigma_0^\prime}{\sigma_0^3} \left[ \rho^2
      \ip{K1}{\indic{[0,t]}} + \ip{K\indic{[0,t]}}{1} - 3 \rho^2 t \ip{K1}{1}
    \right].
  \end{align*}
\end{theorem}

\begin{remark}[Non-Markovian transversality]
In the RL-fBM case, $K\left( t,s\right) = \sqrt{2H} \left\vert t-s\right\vert
^{\gamma }$ with $\gamma =H-1/2$ one computes 
\[
\left\langle 1,K1_{\left[ 0,t\right] }\right\rangle =\frac{1}{\left(
1+\gamma \right) \left( 2+\gamma \right) }\left\{ 1-\left( 1-t\right)
^{2+\gamma }\right\} \in C^{1}\left[ 0,1\right] .
\]%
Interestingly, the transversality condition known from the Markovian setting
($q_{1}=0$, which readily translates to $\dot{f}_{1}^{x}=0$ there) remains
valid here (for $\rho = 0$), at least to order $x^{2}$, in the sense that 
\[
\dot{f}_{t}^{x}\approx \beta _{t}\frac{x^{2}}{2}=\left( \text{const}\right)
\left( 1-t\right) ^{1+\gamma }|_{t=1}=0
\]
\end{remark}

\begin{proof}[Proof of Theorem~\ref{thr:f^x-expansion-general-rho}]

\textbf{First order expansion:}\\
Up to the order needed in order to get the first order term, we have
\begin{align*}
f_t^x&=\alpha_tx+\mathcal{O}(x^2),\\
\dot{f_t}^x&=\dot{\alpha_t}x+\mathcal{O}(x^2),\\
\sigma(K\dot{f}^x)&= \sigma_0 + \sigma'_0 K\dot{\alpha}\ x + \mathcal{O}(x^2),\\
\sigma'(K\dot{f}^x)&= \sigma_0'+\sigma''_0 K\dot{\alpha}\ x
+\mathcal{O}(x^2),\\
F(f^x)&=\langle\sigma^2(K \dot{f}^x),1\rangle \\
&= \sigma^2_0 + \mathcal{O}(x),\\
G(f^x)&=\langle \sigma(K\dot{f}^x),\dot{f}^x\rangle\\
&=\ip{\sigma_0}{\dot{\alpha}} x + \mathcal{O}(x^2).
\end{align*}
Therefore,
\begin{align*}
\langle \sigma(K\dot{f}^x), \indic{[0,t]} \rangle&=
\sigma_0 t + \mathcal{O}(x),\\
\langle \sigma'(K\dot{f}^x) \dot{f}^x,K \indic{[0,t]} \rangle&= \mathcal{O}(x),\\
\langle\sigma \sigma'(K\dot{f}^x), K\indic{[0,t]} \rangle&= \mathcal{O}(1),\\
x-\rho G(f^x)&= (1-\rho \sigma_0 \alpha_1)x + \mathcal{O}(x^2),\\
(x-\rho G(f^x))^2& = \mathcal{O}(x^2).
\end{align*}
This yields for the first order term in \eqref{eq:frelation}
\begin{equation*}
  \alpha_t = \f{\rho( 1-\rho \sigma_0 \alpha_1)}{\barrho^2 \sigma_0} t. 
\end{equation*}
Setting $t = 1$, we get
\begin{equation*}
  \alpha_1 = \f{\rho}{\barrho^2 \sigma_0} - \f{\rho^2}{\barrho^2} \alpha_1,
\end{equation*}
which is solved by $\alpha_1 = \f{\rho}{\sigma_0}$. Inserting this term back
into the equation for $\alpha_t$, we get
\begin{equation}
  \label{eq:alpha_t}
  \alpha_t = \f{\rho}{\sigma_0} t.
\end{equation}

\noindent
\textbf{Second order expansion:}

Using~\eqref{eq:alpha_t} and the ansatz $f^x_t = \alpha_t x +
\half \beta_t x^2 + \mathcal{O}(x^3)$, we re-compute the relevant terms appearing
in the \eqref{eq:frelation}. We have
\begin{equation*}
  \sigma(K \dot{f}^x(s)) = \sigma_0 + \sigma_0^\prime \f{\rho}{\sigma_0}
  (K1)(s) x + \mathcal{O}(x^2)\\
\end{equation*}
and analogously for $\sigma$ replaced by $\sigma^\prime$,
$\sigma\sigma^\prime$. This implies
\begin{gather*}
  \ip{\sigma(K\dot{f}^x)}{\indic{[0,t]}} = \sigma_0 t + \sigma_0^\prime
  \f{\rho}{\sigma_0} \ip{K1}{\indic{[0,t]}} x + \mathcal{O}(x^2),\\
  \ip{\sigma^\prime(K\dot{f}^x) \dot{f}^x}{K \indic{[0,t]}} = \rho
  \f{\sigma^\prime}{\sigma_0} \ip{K\indic{[0,t]}}{1} x + \mathcal{O}(x^2),\\
  \ip{\sigma \sigma^\prime(K\dot{f}^x)}{K\indic{[0,t]}} =
  \sigma_0\sigma_0^\prime \ip{K\indic{[0,t]}}{1} + \mathcal{O}(x).
\end{gather*}
Using the notation introduced earlier, we have
\begin{gather*}
  F(f^x) = \sigma_0^2 + 2 \sigma_0^\prime \rho \ip{K1}{1}x + \mathcal{O}(x^2),\\
  G(f^x) = \rho x + \left( \half \sigma_0 \beta_1 
    + \rho^2
    \f{\sigma_0^\prime}{\sigma_0^2} \ip{K1}{1} \right) x^2 + \mathcal{O}(x^3).
\end{gather*}
This directly implies
\begin{gather*}
  x - \rho G(f^x) = \barrho^2 x - \rho \left( \half \sigma_0 \beta_1 + \rho^2
    \f{\sigma_0^\prime}{\sigma_0^2} \ip{K1}{1} \right) x^2 + \mathcal{O}(x^3), \\
  \left( x - \rho G(f^x) \right)^2 = \barrho^4 x^2 - 2 \barrho^2 \rho \left(
    \half \sigma_0 \beta_1 + \rho^2 \f{\sigma_0^\prime}{\sigma_0^2} \ip{K1}{1}
  \right) x^3 + \mathcal{O}(x^4).
\end{gather*}

We next compute some auxiliary terms appearing in~\eqref{eq:frelation}.
\begin{align*}
  N_1 &\coloneqq \rho (x - \rho G(f^x)) \left(
        \ip{\sigma(K\dot{f}^x)}{\indic{[0,t]}} +
        \ip{\sigma^\prime(K\dot{f}^x) \dot{f}^x}{K\indic{[0,t]}} \right)\\
  &= \rho \barrho^2 \sigma_0 t x + \biggl[\rho^2 \barrho^2
    \f{\sigma_0^\prime}{\sigma_0} \left( \ip{K1}{\indic{[0,t]}} +
    \ip{K\indic{[0,t]}}{1} \right) \\
  &\quad  - \rho^4 \f{\sigma_0^\prime}{\sigma_0} t
    \ip{K1}{1} - \half \rho^2 \sigma_0^2 t \beta_1 \biggr] x^2 + \mathcal{O}(x^3)
\end{align*}
The corresponding denominator is $\barrho^2 F(f^x)$. Using the formula
\begin{equation*}
  \f{a_1 x + a_2 x^2 + \mathcal{O}(x^3)}{b_0 + b_1 x + \mathcal{O}(x^2)} = \f{a_1}{b_0} x + \f{a_2
  b_0 - a_1b_1}{b_0^2} x^2 + \mathcal{O}(x^3),
\end{equation*}
we obtain
\begin{multline}
  \label{eq:N1/D1}
  \f{N_1}{\barrho^2 F(f^x)} = \f{\rho}{\sigma_0} tx + \biggl[ \rho^2
  \f{\sigma_0^\prime}{\sigma_0^3} \left( \ip{K1}{\indic{[0,t]}}  +
    \ip{K\indic{[0,t]}}{1} \right) \\
  -\left(\f{\rho^4}{\barrho^2} + 2 \rho^2 \right)
 \f{\sigma_0^\prime}{\sigma_0^3} t\ip{K1}{1}
    - \half
  \f{\rho^2}{\barrho^2} \beta_1 t
 \biggr] x^2 + \mathcal{O}(x^3)
\end{multline}

For the second term in~\eqref{eq:frelation}, let
\begin{equation*}
  N_2 \coloneqq \left(x - \rho G(f^x) \right)^2
  \ip{(\sigma\sigma^\prime)(K\dot{f}^x)}{K\indic{[0,t]}} = \barrho^4 \sigma_0
  \sigma_0^\prime \ip{K\indic{[0,t]}}{1} x^2 + \mathcal{O}(x^3).
\end{equation*}
The corresponding denominator is $\barrho^2 F(f^x)^2 = \barrho^2 \sigma_0^4 +
\mathcal{O}(x)$. Hence,
\begin{equation}
  \label{eq:N2/D2}
  \f{N_2}{\barrho^2 F(f^x)^2} = \barrho^2
    \f{\sigma_0^\prime}{\sigma_0^3} \ip{K\indic{[0,t]}}{1} x^2 + \mathcal{O}(x^3).
\end{equation}
Combining~\eqref{eq:N1/D1} and~\eqref{eq:N2/D2}, we get
\begin{multline*}
  f_t^x = \f{\rho}{\sigma_0} tx + \biggl[ \rho^2
  \f{\sigma_0^\prime}{\sigma_0^3} \left( \ip{K1}{\indic{[0,t]}} +
    \ip{K\indic{[0,t]}}{1} \right) 
  -\f{\rho^4}{\barrho^2} \f{\sigma_0^\prime}{\sigma_0^3} t
    \ip{K1}{1}  \\- \half
  \f{\rho^2}{\barrho^2} \beta_1 t
 - 2 \rho^2 \f{\sigma_0^\prime}{\sigma_0^3} t \ip{K1}{1} +
   \barrho^2 
 \f{\sigma_0^\prime}{\sigma_0^3} \ip{K\indic{[0,t]}}{1} \biggr] x^2 + \mathcal{O}(x^3)
\end{multline*}
We shall next compute $\beta_1$. Taking the second order terms on both sides
and letting $t = 1$, we obtain
\begin{multline*}
  \half \beta_1 = \rho^2
  \f{\sigma_0^\prime}{\sigma_0^3} 2\ip{K1}{1}
  -\f{\rho^4}{\barrho^2} \f{\sigma_0^\prime}{\sigma_0^3}
    \ip{K1}{1} \\- \half
  \f{\rho^2}{\barrho^2} \beta_1
  - 2 \rho^2 \f{\sigma_0^\prime}{\sigma_0^3} \ip{K1}{1} +
    \barrho^2 \f{\sigma_0^\prime}{\sigma_0^3} \ip{K1}{1}.
\end{multline*}
Moving $\beta_1$ to the other side with $1 + \f{\rho^2}{\barrho^2} =
\f{1}{\barrho^2}$ and collecting terms on the right hand side, we arrive at
\begin{equation*}
  \half \f{1}{\barrho^2} \beta_1 = \f{\sigma_0^\prime}{\sigma_0^3} \ip{K1}{1}
  \left( 2 \rho^2 - \f{\rho^4}{\barrho^2} - 2 \rho^2
    + \barrho^2 \right) = 
  \f{1-2\rho^2}{\barrho^2} \f{\sigma_0^\prime}{\sigma_0^3}
  \ip{K1}{1} 
\end{equation*}
We conclude that
\begin{equation*}
  \beta_1 = 2(1-2\rho^2) \f{\sigma_0^\prime}{\sigma_0^3} \ip{K1}{1}
\end{equation*}
Hence, we obtain
\begin{equation*}
  \beta_t = 2 \f{\sigma_0^\prime}{\sigma_0^3} \left[ \rho^2
    \ip{K1}{\indic{[0,t]}} + \ip{K\indic{[0,t]}}{1} - 3 \rho^2 t \ip{K1}{1}
  \right].\qedhere
\end{equation*}

\end{proof}

\subsubsection{Energy expansion in the general case}
\label{sec:energy-expans-gener}

Now we compute the Taylor expansion of $I(x)$ as defined in
Proposition~\ref{prop:fractional-LDP}. We start with the second term. Plugging
in the optimal path $f_t^x = \alpha_t x + \half \beta_t x^2 + \mathcal{O}(x^3)$ (and
using $\ip{\dot{\beta}}{1} = \beta_1$ as $\beta_0 = 0$) we
obtain
\begin{equation*}
  \half \ip{\dot{f}^x}{\dot{f}^x} = \half \f{\rho^2}{\sigma_0^2} x^2 + \half
  \f{\rho}{\sigma_0} \beta_1 x^3 + \mathcal{O}(x^4).
\end{equation*}

Inserting $\beta_1 = 2(1-2\rho^2) \f{\sigma_0^\prime}{\sigma_0^3} \ip{K1}{1}$
into the above formula for $\left( x - \rho G(f^x) \right)^2$, we get
\begin{equation*}
\left( x - \rho G(f^x) \right)^2 = \barrho^4 x^2 - 2 \barrho^4 \rho
  \f{\sigma_0^\prime}{\sigma_0^2} \ip{K1}{1} x^3 + \mathcal{O}(x^4).
\end{equation*}
Recall the denominator
\begin{equation*}
  2 \barrho^2 F(f^x) = 2 \barrho^2 \sigma_0^2 + 4 \barrho^2 \sigma_0^\prime
  \rho \ip{K1}{1}x + \mathcal{O}(x^2). 
\end{equation*}
Using the expansion of a fraction
\begin{equation*}
  \f{a_2 x^2 + a_3 x^3 + \mathcal{O}(x^4)}{b_0 + b_1 x + \mathcal{O}(x^2)} = \f{a_2}{b_0} x^2 +
  \f{a_3 b_0 - a_2 b_1}{b_0^2} x^3 + \mathcal{O}(x^4),
\end{equation*}
we obtain from
\begin{align*}
  \f{\left( x - \rho G(f^x) \right)^2}{2 \barrho^2 F(f^x)} 
  &=  \f{\barrho^4}{2 \barrho^2 \sigma_0^2} x^2 +\\
  &\quad+
    \f{\left(- 2 \barrho^4 \rho \f{\sigma_0^\prime}{\sigma_0^2} \ip{K1}{1}
    \right) 2 \barrho^2 \sigma_0^2 - \barrho^4\left( 4 \barrho^2 \sigma_0^\prime
    \rho \ip{K1}{1} \right)}{4 \barrho^4 \sigma_0^4} x^3 + \mathcal{O}(x^4)\\
  &= \f{\barrho^2}{2 \sigma_0^2} x^2 
    -2 \barrho^2 \rho\f{ \sigma_0^\prime}{ \sigma_0^4} \ip{K1}{1} x^3 + \mathcal{O}(x^4).
\end{align*}
We note that
\begin{equation*}
  \half
  \f{\rho}{\sigma_0} \beta_1  
    -2 \barrho^2 \rho\f{ \sigma_0^\prime}{ \sigma_0^4} \ip{K1}{1} 
    = 
    \left(  (1-2\rho^2) - 2 (1-\rho^2) \right)\rho
    \f{\sigma_0^\prime}{\sigma_0^4} \ip{K1}{1}
    = 
    -\rho
    \f{\sigma_0^\prime}{\sigma_0^4} \ip{K1}{1}.
\end{equation*}

Adding both terms, we arrive at the
\begin{proposition}
  \label{prop:energy-expansion-general}
  The energy expansion to third order gives
  \begin{equation*}
    I(x) = \f{1}{2\sigma_0^2} x^2 - \rho
    \f{\sigma_0^\prime}{\sigma_0^4} \ip{K1}{1} x^3 + \mathcal{O}(x^4).
  \end{equation*}
\end{proposition}

\subsubsection{Energy expansion for the Riemann-Liouville kernel}
\label{sec:energy-expans-riem}

Let us specialize the energy expansion 
given in Proposition~\ref{prop:energy-expansion-general} 
for the Riemann-Liouville
fBm. Choose $\gamma = H - \half$ and recall that the kernel $K$ takes the form
$K(t,s) = (t-s)^\gamma$. We get
\begin{equation*}
  (K1)(t) = \int_0^t K(t,s) ds = \int_0^t (t-s)^\gamma ds =
  \f{t^{1+\gamma}}{1+\gamma}.
\end{equation*}
The key term $\ip{K1}{1}$ appearing in the energy expansion now gives
\begin{equation*}
  \label{eq:<K1,1>}
  \ip{K1}{1} = \int_0^1 (K1)(t) dt = \int_0^1 \f{t^{1+\gamma}}{1+\gamma} dt =
  \f{1}{(1+\gamma)(2+\gamma)} = \f{1}{(H+1/2)(H+3/2)}.
\end{equation*}
Plugging formula~\eqref{eq:<K1,1>} into the energy expansion,
we obtain the energy expansion for the Riemann-Liouville fractional Browian
motion
\begin{equation*}
  \label{eq:energy-expansion-RL}
  I(x) = \f{1}{2 \sigma_0^2} x^2 - \f{\rho}{(H+1/2)(H+3/2)}
  \f{\sigma_0^\prime}{\sigma_0^4} x^3 + \mathcal{O}(x^4).
\end{equation*}

For completeness, let us also fully describe the time-dependence of the second
order term $\beta_t$ in the expansion of the optimal trajectory
$f^x_t$. Unlike the first order time, here we do not have a linear movement
any more. Indeed
\begin{gather}
  \label{eq:<K1,1[0,t]>}
  \ip{K1}{\indic{[0,t]}} = \int_0^t (K1)(s) ds = \int_0^t
  \f{s^{1+\gamma}}{1+\gamma} ds = \f{t^{2+\gamma}}{(1+\gamma)(2+\gamma)},\\
  \label{eq:<K1[0,t],1>}
  \ip{K\indic{[0,t]}}{1} = \f{1}{(1+\gamma)(2+\gamma)} \left( 1 -
    (1-t)^{2+\gamma} \right).
\end{gather}


\section{Proof of the pricing formula}
\label{sec:price-expansion}

Fix $x\geq 0$ and $\widehat{x}=\frac{\varepsilon }{\hat{\varepsilon}}x$ where $%
\varepsilon =t^{1/2}$ and $\hat{\varepsilon}=t^{H}=\varepsilon ^{2H}$. We
have
\begin{align*}
c(\widehat{x},t) &=E\left( \exp \left( X_{t}\right) -\exp \widehat{x}\right) ^{+} \\
&=E\left( \exp \left( X_{1}^{\varepsilon }\right) -\exp \widehat{x}\right) ^{+}\\
&=E\left( \exp \left( \frac{\varepsilon }{\hat{\varepsilon}}\widehat{X}%
_{1}^{\varepsilon }\right) -\exp \left( \frac{\varepsilon }{\hat{\varepsilon}%
}x\right) \right) ^{+}
\end{align*}
where we recall
\begin{equation*}
\widehat{X}_{1}^{\varepsilon }\equiv \frac{\hat{\varepsilon}}{\varepsilon }%
X_{1}^{\varepsilon }=\int_{0}^{1}\sigma (\hat{\varepsilon}\widehat{B})\widehat{%
\varepsilon}d\left( \bar{\rho}W+\rho B\right) -\half
\varepsilon \hatepsilon \int_0^1 \sigma\left(\hatepsilon \hat{B}_t
\right)^2dt.
\end{equation*}

Consider a Cameron-Martin perturbation of $\widehat{X}_{1}^{\varepsilon }$.
That is, for a Cameron-Martin path $\mathrm{h}=(h,f)\in H_0^1 \times H_0^1$
consider a measure change corresponding to a transformation
$\widehat{\varepsilon}\left( W,B\right) \rightsquigarrow
\widehat{\varepsilon}\left( W,B\right) +\left( h,f\right)$ (transforming the
Brownian motions to Brownian motions with drift), we obtain the Girsanov
density
\begin{align}\label{eq:Girsanovfactor}
G_{\varepsilon } =\exp\left(-\frac{1}{\hat{\varepsilon}}\int_0^1 \dot{h}_sdW_s-\frac{1}{
\hat{\varepsilon}}\int_0^1 \dot{f}_sdB_s-\frac{1}{2\hat{\varepsilon}^{2}}\int_0^1 \left( 
\dot{h}^{2}_s+\dot{f}^{2}_s\right) ds\right).
\end{align}
Under the new measure, $\widehat{X}_{1}^{\varepsilon }$ becomes
$\widehat{Z}_{1}^{\varepsilon }$, where
\begin{equation*}
\widehat{Z}_{1}^{\varepsilon }=\int_{0}^{1}\sigma (\hat{\varepsilon}\widehat{B}_t+\widehat{%
f}_t)\left[ \hat{\varepsilon}d\left( \bar{\rho}W_t+\rho B_t\right) +d\left( \bar{
\rho}h_t+\rho f_t\right) \right] -\half \varepsilon \hatepsilon
\int_0^1 \sigma(\hatepsilon \hat{B}_t + \hat{f}_t)^2 dt.  
\label{ZepsHat}
\end{equation*}

\begin{definition}\label{def:CheapestControl}
  For fixed $x\geq 0$, write $\left( h,f\right) \in \mathcal{K}^{x}$ if
  $\Phi _{1}\left( h, f, \hat{f} \right) =x$. Call such $\left( h,f\right) $
  admissible for arrival at log-strike $x$. Call $\left( h^{x},f^{x}\right) $
  the cheapest admissible control, which attains
  \begin{equation*}
    I\left(x\right) = \inf_{h,f\in H_{0}^1}\left\{
      \frac{1}{2}\int_{0}^{1}\dot{h}^{2}dt+\frac{1}{2}
      \int_{0}^{1}\dot{f}^{2}dt:\Phi _{1}\left( h,f, \hat{f} \right) =x\right\}, 
  \end{equation*}
  where we recall that $\hat{f} = K\dot{f}$ and
  \begin{equation*}
    \Phi _{1}(h, f, \widehat{f}) = \int_{0}^{1}\sigma (\widehat{f})d\left(
      \bar{\rho}h + \rho f\right)
    .
  \end{equation*}
\end{definition}

For any Cameron-Martin path $(h,f)$, the perturbed random variable $\widehat{Z}_{1}^{\varepsilon }$ admits a
stochastic Taylor expansion with respect to $\hat{\varepsilon}$. 

\begin{lemma}\label{Lem:STEforZ}
Fix $\left( h,f\right) \in \mathcal{K}^{x}$ and define $\widehat{Z}%
_{1}^{\varepsilon }$ accordingly. Then 
\begin{equation}\label{eq:ExpforZ}
\widehat{Z}_{1}^{\varepsilon }=x+\hat{\varepsilon}g_{1} +
\hat{\varepsilon}^{2}R_{2}^\varepsilon,
\end{equation}
where $g_{1}$ is a Gaussian random variable, given explicitly by%
\begin{equation}\label{g1Expl}
g_{1}=\int_{0}^{1}\{\sigma (\widehat{f}_{t})d\left( \bar{\rho}W_{t}+\rho
B_{t}\right) +\sigma ^{\prime }(\widehat{f}_{t})\widehat{B}_{t}d\left( \bar{\rho}%
h_{t}+\rho f_{t}\right) \}
,  
\end{equation}
and
\begin{multline}
  \label{eq:R2-def}
  R_2^\varepsilon = \int_0^1 \sigma'\left( \hat{f}_t \right) \hat{B}_t d\left( \barrho W_t
    + \rho B_t \right) 
   - \half \f{\varepsilon}{\hatepsilon} \int_0^1
    \sigma(\hatepsilon \hat{B}_t + \hat{f}_t)^2 dt\\
  + \f{1}{2 \hatepsilon^2} \int_0^{\hatepsilon} \int_0^1
  \sigma''\left( \zeta \hat{B}_t + \hat{f}_t \right) \hat{B}_t^2 \left[
    \hatepsilon d\left( \barrho W_t + \rho B_t \right) + d\left( \barrho h_t +
      \rho f_t \right) 
  \right] \left(
    \hatepsilon - \zeta \right) d\zeta. 
\end{multline}
\end{lemma}

\begin{proof}
By a stochastic Taylor expansion for the controlled process $\widehat{Z}_{t}^{\varepsilon}$  with control $(h,f) \in \mathcal{K}^x$ as in Definition \ref{def:CheapestControl}
and thanks to $\sigma \in C^{2}$, we have at $t=1$
\begin{align*}
\widehat{Z}_{1}^{\varepsilon } &=\int_{0}^{1}\sigma (\hat{\varepsilon}\widehat{B}+%
\widehat{f})\left[ \hat{\varepsilon}d\left( \bar{\rho}W+\rho B\right) +d\left( 
\bar{\rho}h+\rho f\right) \right] -\half \varepsilon \hatepsilon
\int_0^1 \sigma(\hatepsilon \hat{B}_t + \hat{f}_t)^2 dt\\
&=\int_{0}^{1}\sigma (\widehat{f})d\left( \bar{\rho}h+\rho f\right)
+\hat{\varepsilon}\int_{0}^{1}\{\sigma (\widehat{f})d\left( \bar{\rho}W+\rho
B\right) +\sigma ^{\prime }(\widehat{f}) \widehat{B}d\left( \bar{\rho}h+\rho f\right) \}
+\\
&\quad+\hatepsilon^2\int_0^1 \sigma'\left( \hat{f}_t \right) \hat{B}_t d\left( \barrho W_t
    + \rho B_t \right) 
   - \half \varepsilon \hatepsilon \int_0^1
    \sigma(\hatepsilon \hat{B}_t + \hat{f}_t)^2 dt\\
  &\quad + \half \int_0^{\hatepsilon} \int_0^1
  \sigma''\left( \zeta \hat{B}_t + \hat{f}_t \right) \hat{B}_t^2 \left[
    \hatepsilon d\left( \barrho W_t + \rho B_t \right) + d\left( \barrho h_t +
      \rho f_t \right) 
  \right] \left(
    \hatepsilon - \zeta \right) d\zeta.
\end{align*}
Collecting terms in powers of $\hatepsilon$ and with the random variable $g_1$
as in \eqref{g1Expl} (recalling that $\hat{\varepsilon} \varepsilon \in \mathcal{O}(\hat{\varepsilon}^2)$), we have
\begin{align*}
\widehat{Z}_{1}^{\varepsilon } =\int_{0}^{1}\sigma (\widehat{f})d\left(
  \bar{\rho}h+\rho f\right) +\hat{\varepsilon}g_1+\mathcal{O}(\hat{\varepsilon}^2), 
\end{align*}
furthermore, since $(h,f)\in \mathcal{K}^x$, by the definition of $\Phi_1$, it
holds that
\begin{align*}
\int_{0}^{1}\sigma (\widehat{f})d\left( \bar{\rho}h+\rho f\right)=x.
\end{align*}
This proves the statement \eqref{eq:ExpforZ} and the statement that $g_1$ is
Gaussian is immediate from the form \eqref{g1Expl}.
\end{proof}

\noindent Finally, we determine an explicit form of the Girsanov
  density $G_\varepsilon$ for the choice where $(h^x, f^x)$  in
  \eqref{eq:Girsanovfactor} are chosen the cheapest admissible control
  (cf. Definition \ref{def:CheapestControl}. 
  Similarly to classical works of Azencott, Ben Arous and others, see, for
  instance,~\cite{BA88}, we show that
  the stochastic integrals in the exponent of $G_\varepsilon$ are proportional
  to the first order term $g_1$ (with factor $I^\prime(x)$) \emph{when
    evaluated at the minimizing configuration} $(h^x, f^x)$.

\begin{lemma}
We have%
\begin{equation*}
\int_0^1 \dot{h}^{x}_tdW_t+\int_0^1 \dot{f}^{x}_tdB_t = I^{\prime }\left( x\right) g_{1}.
\end{equation*}
\end{lemma}

\begin{proof}
  See Lemma~\ref{lem:appendix-2}.
\end{proof}

\noindent With these preparations in place, we are now ready to prove the pricing formula from
Section~\ref{sec:main-results}.

\begin{proof}[Proof of Theorem~\ref{thr:main-price=expansion}]
  With a Girsanov factor (all integrals on $\left[ 0,1\right] $)
  \begin{equation*}
    G_{\varepsilon } = e^{-\frac{1}{\hat{\varepsilon}}\int \dot{h}dW-\frac{1}{%
                         \hat{\varepsilon}}\int \dot{f}dB-\frac{1}{2\hat{\varepsilon}^{2}}\int \left( 
                         \dot{h}^{2}+\dot{f}^{2}\right) dt}
  \end{equation*}
  and (evaluated at the minimizer)
\begin{equation*}
G_{\varepsilon }|_{\ast } =e^{-\frac{I\left( x\right) }{\hat{\varepsilon}%
                                  ^{2}}}e^{-\frac{I^{\prime }\left( x\right) g_{1}\left( \omega \right) }{\widehat{%
                                  \varepsilon}}},
\end{equation*}
  we have, setting
  $\widehat{U}^{\varepsilon } \coloneqq \widehat{Z}_{1}^{\varepsilon }-x=\widehat{ \varepsilon}g_{1} +\hat{\varepsilon}^{2}R_{2}^\varepsilon$
  \begin{align*}
    c(\widehat{x},t) &=E\left[ \left( \exp \left( \frac{\varepsilon }{\widehat{%
                         \varepsilon}}\widehat{Z}_{1}^{\varepsilon }\right) -\exp \left( \frac{%
                         \varepsilon }{\hat{\varepsilon}}x\right) \right) ^{+}G_{\varepsilon }|_{\ast
                         }\right] \\
                     &=e^{\frac{\varepsilon }{\hat{\varepsilon}}x}E\left[ \left( \exp \left( 
                         \frac{\varepsilon }{\hat{\varepsilon}}\widehat{U}^{\varepsilon }\right)
                         -1\right) ^{+}G_{\varepsilon }|_{\ast }\right] \\
                     &=e^{-\frac{I\left( x\right) }{\hat{\varepsilon}^{2}}}e^{\frac{\varepsilon 
                         }{\hat{\varepsilon}}x}E\left[ \left( \exp \left( \frac{\varepsilon }{\widehat{%
                         \varepsilon}}\widehat{U}^{\varepsilon }\right) -1\right) ^{+}e^{-\frac{I^{\prime
                         }\left( x\right) g_{1} }{\hat{\varepsilon}}}\right] \\
                     &=e^{-\frac{I\left( x\right) }{\hat{\varepsilon}^{2}}}e^{\frac{\varepsilon 
                         }{\hat{\varepsilon}}x}E\left[ \left( \exp \left( \frac{\varepsilon }{\widehat{%
                         \varepsilon}}\widehat{U}^{\varepsilon }\right) -1\right) e^{-\frac{I^{\prime
                         }\left( x\right) }{\hat{\varepsilon}^{2}}\widehat{U}^{\varepsilon
                         }}e^{I^{\prime }\left( x\right) R_{2}} \indic{\widehat{U}^{\varepsilon }\geq 0} \right] .
    \\
                     &=e^{-\frac{I\left( x\right) }{\hat{\varepsilon}^{2}}}e^{\frac{\varepsilon 
                         }{\hat{\varepsilon}}x}J\left( \varepsilon
                       ,x\right). \qedhere
  \end{align*}
\end{proof}

\section{Proof of the moderate deviation expansions}
\label{sec:moderate-deviations}

In Section 2, we pointed out that (iiic) is exactly what one get from (call price) large deviations (\ref{equ:LDPu}), if heuristically applied to $x \varepsilon^{2\beta}$. We now sketch a proper derivation based on moderate deviations.

\begin{lemma} \label{lem:iiiabc}
Assume (iiia-b) from Assumption \ref{ass:sigma}. Then an upper moderate deviation estimate holds both for calls and digital calls. That is, we have 
\begin{enumerate}
\item[(iiic)] For every $\beta \in (0,H)$,  and every fixed $x>0$, and $\hat x_\varepsilon := x \varepsilon^{1-2H+2\beta}$,
  $$
          E [  (e^{X^\varepsilon_1} - e^{\hat x_\varepsilon })^+ ]    
      \le \exp \left( - \frac{x^2 + o(1)}{2 \sigma_0^2 \varepsilon^{4H - 4 \beta}}  \right)
 $$  
\end{enumerate}
and also
\begin{equation}
P [  X^\varepsilon_1 > \hat x_\varepsilon ]    
      \le \exp \left( - \frac{x^2 + o(1)}{2 \sigma_0^2 \varepsilon^{4H - 4 \beta}}  \right).   \label{e:showthisfirst}
 \end{equation}   
\end{lemma} 
\begin{proof} (Sketch) Recall $\sigma(.)$ smooth but unbounded and recall $\hat x_\varepsilon := x \varepsilon^{1-2H+2\beta}$. In case of $\beta=0$ and $H=1/2$ a large deviation principle (LDP) for $(X^\varepsilon_1 \hat \varepsilon / \varepsilon)$ is readily reduced, via exponential equivalence, to a LDP for the family of stochastic It\^o
 integrals given by $\int \sigma( \hat \varepsilon \hat B)  \hat \varepsilon dZ$ for some Brownian $Z$, $\rho$-correlated with $B$. There are then many ways to establish a LDP 
 for this family. A particularly convenient one, that requires no growth restriction on $\sigma$, uses continuity of stochastic integration with respect to the rough path $(B,Z,\int B dZ)=(B,Z,\int \hat B dZ)$ in suitable  metrics, for which a LDP is known \cite[Ch 9.3]{friz2014course}. It was
 pointed out in \cite{BFGMS17} that a similar reasoning is possible when
 $H<1/2$, the rough path is then replaced by a ``richer enhancement'' of
 $(B,Z)$, the precise size of which depends on $H$, for which again one has a
 LDP. 
 A moderate deviation priniple (MDP) for $(X^\varepsilon_1 \hat \varepsilon / \varepsilon)$ is a LDP for $(\varepsilon^{-2\beta} X^\varepsilon_1 \hat \varepsilon / \varepsilon)$ for 
 $\beta \in (0,H)$. This can be reduced to a LDP, with $\bar \varepsilon := \varepsilon^{-2\beta} \hat \varepsilon = \varepsilon^{2H-2\beta}$, for
 $$ \varepsilon^{-2\beta} \int_0^1 \sigma( \hat \varepsilon \hat B)  \hat \varepsilon dZ =   \int_0^1 \sigma( \hat \varepsilon \hat B)  \bar \varepsilon dZ \equiv   \int_0^1 \sigma_\varepsilon ( \bar \varepsilon \hat B)  \bar \varepsilon dZ$$
 with speed $\bar \varepsilon^2  $. Also, $\sigma_\varepsilon ( \cdot) \equiv \sigma (\varepsilon^{2\beta} \cdot )$ convergens (with all derivatives) locally uniformly to the constant 
 function $\sigma_0$, and one checks that  $ \varepsilon^{-2\beta} \int_0^1 \sigma( \hat \varepsilon \hat B) $ is exponentially equivalent to the (Gaussian) family given by $\sigma_0 
 \bar \varepsilon Z_1$, with law $\mathcal{N} (0, \sigma_0^2 \bar \varepsilon^2) = \mathcal{N} (0, \sigma_0^2 \varepsilon^{4H-4\beta})  $ which gives (\ref{e:showthisfirst}), even with equality.  (Showing this exponential equivalence can again be done for $\sigma$ without growth restrictions.) 

%
%
%
   
   We have not yet used either assumption (iiia-b). These become important in order to extend estimate (\ref{e:showthisfirst}) to the case of genuine call payoffs. We can follow here a well-known argument (Forde-Jacquier, Pham, ...) with the ``moderate'' caveat to carry along a factor $\varepsilon^{2\beta}$. In fact, this is close in spirit to what already happens with rough volatility where one has to carry along a factor $ \hat \varepsilon / \varepsilon = \varepsilon^{2H-1}$. The remaining details then follow essentially ``Appendix C. Proof of Corollary 4.13., part (ii) upper bound'' of  \cite{FZ17}, noting perhaps that the authors use their assumptions to show validity of what we simply assumed as condition (iiib), and also that one works with the quadratic rate function $I''(0)x^2 = \frac{x^2}{2 \sigma_0^2}$ throughout.
  \end{proof} 

\begin{remark} By an easy argument similar to ``Appendix C. Proof of Corollary 4.13., part (i) lower bound'' of  \cite{FZ17} one sees that validity of the call price upper bound (iiic) implies the corresponding digital call price upper bound (\ref{e:showthisfirst}.) For this reason, we only emphasized (iiic) but not (\ref{e:showthisfirst}) in Section 2.
\end{remark} 
  
%
%

In a classical work Azencott \cite{Az82} (see also \cite{Aze85}, \cite[Th\'eor\`eme 2]{BA88}) obtained asymptotic expansions of functionals of Laplace type
on Wiener space, of the type ``$E[\exp (- F(X^\varepsilon)/\varepsilon^2)]$'', for small noise diffusions $X^\varepsilon$. This refines the large deviation (equivalently: Laplace) principle 
of Freidlin--Wentzell for small noise diffusions. In a nutshell, for fixed $X_0=x$, Azencott gets expansions of the form $e^{- c/\varepsilon^2} (\alpha_0 + \alpha_1 \varepsilon ...)$. His 
ideas (used by virtually all subsequent works in this direction) are a Girsanov transform, to make the minimizing path ``typical'', followed by localization around the minimizer 
(justified by a good large deviation principle), and finally a local (stochastic Taylor) type analysis near the minimizer. None of these ingredients rely on the Markovian structure (or, 
relatedly, PDE arguments). As a consequence (and motivation for this work) such expansions were also obtained in the (non-Markovian) context of rough differential equations driven by fractional Brownian 
motion \cite{inahama2013, BO15} with $H<1/2$. 

\medskip

And yet, our situation is different in the sense that call price Wiener functionals do not fit the form studied by Azencott and others, nor can we in fact expect a similar expansion: Example
\ref{ex:BS-example} gives a Black-Scholes call price expansion of the form constant times $e^{-c\varepsilon^2} (\varepsilon^3 + ...)$. Azencott's ideas are nonetheless very relevant to us: we  already used the Girsanov formula in Theorem {\ref{thr:main-price=expansion} in order to have a tractable expression for $J$. It thus ``only'' remains to carry out the localization and do some local analysis.  We again content ourselves with a sketch and leave full technical details as well as some extensions to a forthcoming technical note.

\begin{proposition} \label{prop:J-bounds}
 In the context of Theorem \ref{thm:mod}, let $x>0$. Then the factor $J$ is negligible in the sense that, for every $\theta > 0$,
 $$
                          \varepsilon^{\theta} \log J ( \varepsilon, x \varepsilon^{2\beta} ) \to 0  \ \ \ \text{ as $\varepsilon \to 0$} \ .     
$$     
\end{proposition}  

\begin{proof} (Sketch). {\it Step 1. Localization} One shows that 
\begin{equation*}
    J\left( \varepsilon ,x\right) :=E\left[ e^{-\frac{I^{\prime }\left( x\right) 
        }{\hat{\varepsilon}^{2}}\widehat{U}^{\varepsilon }}\left( \exp \left( \tfrac{%
            \varepsilon }{\hat{\varepsilon}}\widehat{U}^{\varepsilon }\right) -1\right)
      e^{I^{\prime }\left( x\right) R^\varepsilon_{2}} \ 
      \indic{\widehat{U}^{\varepsilon }\geq 0}\right]
  \end{equation*}%
can be replaced, in the sense that the error $ |J\left( \varepsilon ,x \varepsilon^{2\beta} \right)  - J_\delta \left( \varepsilon ,x \varepsilon^{2\beta} \right)|$ is exponentially small (cf. \cite[Lemme 1.32]{BA88}), with 
\begin{equation*}
    J_\delta \left( \varepsilon ,x\right) :=E\left[ e^{-\frac{I^{\prime }\left( x\right) 
        }{\hat{\varepsilon}^{2}}\widehat{U}^{\varepsilon }}\left( \exp \left( \tfrac{%
            \varepsilon }{\hat{\varepsilon}}\widehat{U}^{\varepsilon }\right) -1\right)
      e^{I^{\prime }\left( x\right) R^\varepsilon_{2}} \ 
      \indic{\widehat{U}^{\varepsilon } \in [0, \hat \delta]    }\right] .
  \end{equation*}%
Unlike the works of \cite{Az82,BA88}, however, this is not a simple consequence of large (or here: moderate) deviation upper estimates alone, but requires the corresponding call price estimate (iiic), as provided by Lemma \ref{lem:iiiabc}.

{\it Step 2. Local analysis.} Recall that $\widehat{U}^{\varepsilon }$ decomposes into a Gaussian random variable $g_1$ and remainder $R^\varepsilon_{2}$. In order to control this remainder {\it without} imposing boundedness assumption on $\sigma(.)$, we can show that it is well concentrated on the relevant parts of the probability space in the sense of a ``localized remainder tail estimate'' (cf. \cite[Prop. 4.3.]{Az82}), of the form
$$
      P\left[ \left\vert R^\varepsilon_{2} \right\vert >r,\left\vert \hat \varepsilon \hat B\right\vert _{\infty ;\left[ 0,1\right] }<\kappa \right] \lesssim e^{-\text{(const)}r^{2}} \ .
$$
(It is in this step that we exploit $C^2$-regularity of $\sigma$, which allows to write the remainder in terms of local martingales, stopped after leaving a $\kappa$-neighbourhood of zero, whose quadratic variation then can be estimated and leads to the claimed tail estimate.) One then estimates $J_\delta$ from above, separately on $\{ \hat \varepsilon | \hat B |_{\infty ;\left[ 0,1\right] }<\kappa \}$ (using the above estimate) and its complement, using Fernique estimates. For the lower bound, 
use again localized remainder tail estimate, plus some elementary calculus estimates of the form
\begin{equation*}
u\mapsto (e^{u}-1)^{\frac{1}{p}}e^{-\frac{I^{\prime }\left( x\right) }
{\hat \varepsilon ^{2}p}u}\geq \text{(const)} u^{1/p},
\end{equation*}%
for some positive constant, and $u / \hat \varepsilon ^{2}$ small enough. 
\end{proof}

\section{Proof of the implied volatility expansion}

\noindent With Theorem \ref{thr:main-price=expansion} in place, we now turn to the proof of the implied volatility expansion, formulated in Theorem \ref{T:stra}.

\begin{proof}[Proof of Theorem \ref{T:stra}]
We will use an asymptotic formula for the dimensionless implied variance
$$
V^2_t=t\sigma_{\rm impl}(k_t,t)^2,\quad t> 0, 
$$
obtained in \cite{GL14}. It follows from the first formula in Remark 7.3 in \cite{GL14} that
\begin{equation}
V_t^2-\frac{k_t^2}{2L_t}=\mathcal{O}\left(\frac{k_t^2}{L_t^2}(k_t+|\log k_t|+\log L_t)\right),\quad t\rightarrow 0,
\label{E:x}
\end{equation}
where $L_t=-\log c(k_t,t)$, $t> 0$.
 
We will need the following formula that was established in the proof of Theorem \ref{thm:mod}: 
\begin{equation}
	L_t=\frac{I(kt^{\beta})}{t^{2H}}+ \mathcal{O}(t^{-\rtheta})
\label{E:ref1}
\end{equation}
as $t\rightarrow 0$, for all $x\ge 0$ and $\beta\in[0,H)$ and
  any $\theta > 0$.
Let us first assume $\frac{2H}{n+1}\le\beta<\frac{2H}{n}$. Using the energy expansion, we obtain from (\ref{E:ref1}) that
\begin{align}
L_t&=\sum_{i=2}^n\frac{I^{(i)}(0)}{i!}k^it^{i\beta-2H}+\mathcal{O}\left(t^{-\rtheta}
\right)
=\frac{I^{\prime\prime}(0)}{2}k^2t^{2\beta-2H} \nonumber \\
&\quad\times\left[1+\sum_{i=3}^n\frac{2I^{(i)}(0)}{i!I^{\prime\prime}(0)}k^{i-2}t^{(i-2)\beta}
+\mathcal{O}\left(t^{2H-2\beta-\rtheta}
\right)\right]
\label{E:ref2}
\end{align}
as $t\rightarrow 0$. The second term in the brackets on the right-hand side of (\ref{E:ref2}) disappears if $n=2$.
\begin{remark}\label{R:pou} \rm
	Suppose $n\ge 2$ and $\frac{2H}{n+1}\le\beta<\frac{2H}{n}$. Then formula (\ref{E:ref2}) is optimal. Next, suppose 
	$n\ge 2$ and $0<\beta<\frac{2H}{n+1}$. In this case, there exists $m\ge n+1$ such that $\frac{2H}{m+1}\le\beta<\frac{2H}{m}$, and hence
	(\ref{E:ref2}) holds with $m$ instead of $n$. However, we can replace $m$ by $n$, by making the error term worse. It is not hard to see that the following formula holds for all $n\ge 2$ and
	$0<\beta<\frac{2H}{n+1}$:
	\begin{align}
	L_t&=\sum_{i=2}^n\frac{I^{(i)}(0)}{i!}k^it^{i\beta-2H}+\mathcal{O}\left(t^{(n+1)\beta-2H}\right)
	=\frac{I^{\prime\prime}(0)}{2}k^2t^{2\beta-2H} \nonumber \\
	&\quad\times\left[1+\sum_{i=3}^n\frac{2I^{(i)}(0)}{i!I^{\prime\prime}(0)}k^{i-2}t^{(i-2)\beta}
	+\mathcal{O}\left(t^{(n-1)\beta}\right)\right]
	\label{E:ref4}
	\end{align}
	as $t\rightarrow 0$ provided we choose $\theta$ small
          enough.
\end{remark}

Let us continue the proof of Theorem \ref{T:stra}. Since $k_t\approx t^{\frac{1}{2}-H+\beta}$ and $L_t\approx t^{2\beta-2H}$ as $t\rightarrow 0$, (\ref{E:x}) implies that
\begin{equation}
V_t^2=\frac{k^2t^{1-2H+2\beta}}{2L_t}+\mathcal{O}\left(t^{1+2H-2\beta-\rtheta}
\right),\quad t\rightarrow 0.
\label{E:xxx}
\end{equation}
Next, using the Taylor formula for the function $u\mapsto\frac{1}{1+u}$, and setting
$$
u=\sum_{i=3}^n\frac{2I^{(i)}(0)}{i!I^{\prime\prime}(0)}k^{i-2}t^{(i-2)\beta}+
\mathcal{O}(t^{2H-2\beta-\rtheta}
),
$$
we obtain from (\ref{E:ref2}) that
\begin{align*}
&(2L_t)^{-1}=\frac{t^{2H-2\beta}}{k^2I^{\prime\prime}(0)}
\left[\sum_{j=0}^{n-2}(-1)^ju^j+\mathcal{O}(u^{n-1})\right]
\end{align*}
as $t\rightarrow 0$. 
It follows from $\frac{2H}{n+1}\le\beta<\frac{2H}{n}$ that $(n-1)\beta\ge 2H-2\beta$, and hence
\begin{align*}
&(2L_t)^{-1}=\frac{t^{2H-2\beta}}{k^2I^{\prime\prime}(0)}
\left[\sum_{j=0}^{n-2}(-1)^ju^j\right]+\mathcal{O}(t^{4H-4\beta-\rtheta}
) \\
&=\frac{t^{2H-2\beta}}{k^2I^{\prime\prime}(0)}
\left[\sum_{j=0}^{n-2}(-1)^j\left(\sum_{i=3}^n\frac{2I^{(i)}(0)}{i!I^{\prime\prime}(0)}
k^{i-2}t^{(i-2)\beta}\right)^j\right]+\mathcal{O}(t^{4H-4\beta-\rtheta}
)
\end{align*}
as $t\rightarrow 0$. Now, (\ref{E:xxx}) gives
\begin{align*}
V_t^2&=\frac{t}{I^{\prime\prime}(0)}
\left[\sum_{j=0}^{n-2}(-1)^j\left(\sum_{i=3}^n\frac{2I^{(i)}(0)}{i!I^{\prime\prime}(0)}
k^{i-2}t^{(i-2)\beta}\right)^j\right]  \\
&\quad+\mathcal{O}\left(t^{1+2H-2\beta-\rtheta}
\right)
\end{align*}
as $t\rightarrow 0$. Finally, by cancelling a factor of $t$ in the previous formula, we obtain
formula (\ref{E:dada}) for $\frac{2H}{n+1}\le\beta<\frac{2H}{n}$. The proof in the case where 
$\beta\le\frac{2H}{n+1}$ is similar. Here we take into account Remark \ref{R:pou}. This completes the proof of Theorem \ref{T:stra}. \end{proof}

\appendix

\section{Auxiliary lemmas}
In this section we provide and prove some auxiliary lemmas, which are used in the preparations to the proof of Theorem \ref{thr:main-price=expansion}. We start with a technical Lemma, that justifies the derivation.
\begin{lemma}
  Assume $\sigma \left( .\right) >0$ and $%
  \left\vert \rho \right\vert <1$. Then $\mathcal{K}^{x}$ is a Hilbert
  manifold near any $\mathfrak{h} \coloneqq (h,f) \in \mathcal{K}^{x} \subset
  \mathfrak{H} \coloneqq H^1_0 \times H^1_0$.
\end{lemma}

\begin{proof}
  Similar to Bismut~\cite[p.~25]{Bis84} we need to show that
  $D\varphi _{1}\left( \mathfrak{h}\right)$ is surjective where
  $\varphi _{1}\left( \mathfrak{h}\right) :$ $%
  \mathfrak{H} \rightarrow \mathbb{R}$ with
\begin{equation*}
\varphi _{1}\left( \mathfrak{h}\right) =\varphi _{1}\left( h,f\right)
=\int_{0}^{1}\sigma (\widehat{f})d\left( \bar{\rho}h+\rho f\right) .
\end{equation*}%
From 
\begin{eqnarray*}
\varphi _{1}\left( \mathfrak{h}+\delta \mathfrak{h}^{\prime }\right)
&=&\int_{0}^{1}\sigma (\widehat{f}+\delta \widehat{f}^{\prime })d\left( \bar{\rho}%
h+\rho f+\delta (\bar{\rho}h^{\prime }+\rho f^{\prime })\right) \\
&=&\varphi _{1}\left( \mathfrak{h}\right) +\delta \int_{0}^{1}\sigma (\widehat{f}%
)d(\bar{\rho}h^{\prime }+\rho f^{\prime }) \\
&&+\delta \int_{0}^{1}\sigma ^{\prime }(\widehat{f})\widehat{f}^{\prime }d\left( 
\bar{\rho}h+\rho f\right) +o\left( \delta \right) .
\end{eqnarray*}%
the functional derivative $D\varphi _{1}\left( \mathfrak{h}\right) $ can be
computed explicitly. In fact, even the computation%
\begin{equation*}
\left( D\varphi _{1}\left( \mathfrak{h}\right) ,\left( h^{\prime },0\right)
\right) =\bar{\rho}\int_{0}^{1}\sigma (\widehat{f})dh^{\prime }
\end{equation*}%
is sufficient to guarantee surjectivity of $D\varphi _{1}\left( \mathfrak{h}%
\right) $.
\end{proof}

We now deliver the proof of Lemma~\ref{Lem:STEforZ}, which determines the form of the Girsanov measure change \eqref{eq:Girsanovfactor} for the minimizing configuration (Definition \ref{def:CheapestControl}), denoted by $(h^x.f^x)\in \mathcal{K}^x$.

\begin{lemma}\label{lem:appendix-2}
(i) Any optimal control $\mathfrak{h}^{0}=\left( h^{x},f^{x}\right) \in 
\mathcal{K}^{x}$ is a critical point of 
\begin{equation*}
\mathfrak{h}=\left( h,f\right) \mapsto -I\left( \varphi _{1}^{\mathfrak{h}%
}\right) +\frac{1}{2}\left\Vert \mathfrak{h}\right\Vert _{\mathfrak{H}}^{2};
\end{equation*}%
(ii) it holds that%
\begin{equation*}
\int_0^1 \dot{h}^{x}dW+\int_0^1 \dot{f}^{x}dB=I^{\prime }\left( x\right) g_{1}.
\end{equation*}
\end{lemma}

\begin{proof}
(Step 1) Write $\mathfrak{h}=\left( h,f\right) $ and 
\begin{equation*}
\varphi _{1}\left( \mathfrak{h}\right) =\varphi _{1}\left( h,f\right)
=\int_{0}^{1}\sigma (\widehat{f})d\left( \bar{\rho}h+\rho f\right) \text{.}
\end{equation*}%
Let $\mathfrak{h}^{0}=\left( h^{x},f^{x}\right) \in \mathcal{K}^{x}$ an optimal control. Then 
\begin{equation*}
\mathfrak{Ker}D\varphi _{1}\left( \mathfrak{h}^{0}\right) =T_{\mathfrak{h}^{0}}%
\mathcal{K}^{x}=\left\{ \mathfrak{h}\in \mathfrak{H}^{1}:D\varphi _{1}\left( 
\mathfrak{h}\right) =0\right\} .
\end{equation*}%
(This requires $\mathcal{K}^{x}$ to be a Hilbert manifold near $\mathfrak{h}%
^{0}$, as was seen in the last lemma.)$\newline
$(Step 2) For fixed $\mathfrak{h}\in \mathfrak{H}$, define 
\begin{equation*}
u\left( t\right) :=-I\left( \varphi _{1}^{\mathfrak{h}^{0}+t\mathfrak{h}}\right)
+\frac{1}{2}\left\Vert \mathfrak{h}^{0}+t\mathfrak{h}\right\Vert _{\mathfrak{H}}^{2}\geq 0
\end{equation*}%
with equality at $t=0$ (since $x=\varphi _{1}^{\mathfrak{h}^{0}}$ and $I\left(
x\right) =\frac{1}{2}\left\Vert \mathfrak{h}^{0}\right\Vert _{\mathfrak{H}}^{2}$) and
non-negativity for all $t$ because $\mathfrak{h}^{0}+t\mathfrak{h}$ is an
admissible control for reaching $\tilde{x}=\varphi _{1}^{\mathfrak{h}^{0}+t%
\mathfrak{h}}$ (so that $I\left( \tilde{x}\right) =\inf \left\{ ...\right\}
\leq \frac{1}{2}\left\Vert \mathfrak{h}^{0}+t\mathfrak{h}\right\Vert _{\mathfrak{H}}^{2}$.)%
\newline
(Step 3) We note that $\dot{u}\left( 0\right) =0$ is a consequence of $u\in
C^{1}$ near $0$, $u\left( 0\right) =0$ and $u\geq 0$. In other words, $%
\mathfrak{h}^{0}$ is a critical point for 
\begin{equation*}
\mathfrak{H}^{1}\ni \mathfrak{h}\mapsto -I\left( \varphi _{1}^{\mathfrak{h}%
}\right) +\frac{1}{2}\left\Vert \mathfrak{h}\right\Vert _{\mathfrak{H}}^{2}.
\end{equation*}%
(Step 4) The functional derivative of this map at $\mathfrak{h}^{0}$ must
hence be zero. In particular, for all $\mathfrak{h}\in \mathfrak{H}$, 
\begin{eqnarray*}
0 &\equiv &-I^{\prime }\left( \varphi _{1}^{\mathfrak{h}^{0}}\right)
\left\langle D\varphi _{1}\left( \mathfrak{h}^{0}\right) ,\mathfrak{h}%
\right\rangle +\left\langle \mathfrak{h}^{0},\mathfrak{h}\right\rangle \\
&=&-I^{\prime }\left( x\right) \left\langle D\varphi _{1}\left( \mathfrak{h}%
^{0}\right) ,\mathfrak{h}\right\rangle +\left\langle \mathfrak{h}^{0},\mathfrak{h}%
\right\rangle .
\end{eqnarray*}%
(Step 5) With $\mathfrak{h}^{0}=\left( h^{x},f^{x}\right) $ and $\mathfrak{h}%
=\left( h,f\right) $ 
\begin{eqnarray*}
\left\langle D\varphi _{1}\left( \mathfrak{h}^{0}\right) ,\mathfrak{h}%
\right\rangle &=&\left. \frac{d}{d\varepsilon }\right\vert _{\varepsilon
=0}\int_{0}^{1}\sigma (\widehat{f}^{x}+\varepsilon \widehat{f})d\left( \bar{\rho}%
h^{x}+\rho f^{x}+\varepsilon \left( \bar{\rho}h+\rho f\right) \right) \\
&=&\int_{0}^{1}\sigma (\widehat{f}^{x})d\left( \bar{\rho}h+\rho f\right)
+\int_{0}^{1}\sigma ^{\prime }(\widehat{f}^{x})\widehat{f}d\left( \bar{\rho}%
h^{x}+\rho f^{x}\right)
\end{eqnarray*}%
By continuous extension, replace $\mathfrak{h}=\left( h,f\right) $ by $\left(
W,B\right) $ above and note that%
\begin{equation*}
\left\langle D\varphi _{1}\left( \mathfrak{h}^{0}\right) ,\left( W,B\right)
\right\rangle =g_{1}
\end{equation*}%
since indeed $g_{1} =\int_{0}^{1} \sigma (\widehat{f}%
_{t})d\left( \bar{\rho}W_{t}+\rho B_{t}\right) +\sigma ^{\prime }(\widehat{f}%
_{t})\widehat{B}_{t}d\left( \bar{\rho}h_{t}+\rho f_{t}\right) $. Hence
\begin{equation*}
\int_0^1 \dot{h}^{x}dW+\int_0^1 \dot{f}^{x}dB=I^{\prime }\left( x\right) g_{1}. \qedhere
\end{equation*}
\end{proof}

\bibliographystyle{alpha}

\end{document}